\documentclass[12pt,draftclsnofoot,onecolumn,journal,letterpaper]{IEEEtran}

\usepackage{fancyhdr}
\usepackage{graphicx}
\usepackage{subfigure}
\usepackage{algorithm}
\usepackage{algorithmic}
\usepackage{overpic}
\usepackage{color}
\usepackage{cite}
\usepackage{amsmath}
\usepackage{amsthm}
\usepackage{amssymb}
\usepackage{bm}
\usepackage{mathrsfs}
\theoremstyle{plain}
\newtheorem{theorem}{Theorem}
\newtheorem{lemma}[theorem]{Lemma}%[chapter]
\newtheorem{remark}[theorem]{Remark}%[chapter]

\definecolor{OXO-emph}{RGB}{153,0,0}
\title{Expansion Coding for Channel and Source Coding}
\author{\IEEEauthorblockN{Hongbo~Si, O.~Ozan~Koyluoglu, Kumar~Appaiah
and Sriram~Vishwanath
\thanks{
 The material in this paper was presented in part at the 2012 IEEE International Symposium on Information Theory, Boston, MA, Jul. 2012, and in part at the 2014 IEEE International Symposium on Information Theory, Honolulu, HI, Jul. 2014.}
\thanks{H. Si and S. Vishwanath are with the Laboratory for Informatics, Networks, and Communications,
Wireless Networking and Communications Group,
The University of Texas at Austin, Austin, TX 78712.
Email: sihongbo@mail.utexas.edu, sriram@austin.utexas.edu.}
\thanks{O. O. Koyluoglu is with the
Department of Electrical and Computer Engineering,
The University of Arizona, Tucson, AZ 85721.
Email: ozan@email.arizona.edu.}
\thanks{K. Appaiah is with the Department of Electrical Engineering, Indian Institute of Technology Bombay. Email: akumar@ee.iitb.ac.in.}
}}

\begin{document}
\maketitle

%%%%%%%%%%%%%%%%%%%%%%%%%%%%%%%%%%%%%%%%%%%%%%%%%%%%%%%%%%%%%%%%%%%%%%%%%%%%%%
%%%%%%%%%%%%%%%%%%%%%%%%%%%%%%%%%%%%%%%%%%%%%%%%%%%%%%%%%%%%%%%%%%%%%%%%%%%%%%
\begin{abstract}
A general method of coding over expansion is proposed,
which allows one to reduce the highly non-trivial problems
of coding over analog channels and compressing analog sources to a set of much simpler subproblems, coding over discrete
channels and compressing discrete sources. More specifically, the focus of this paper is on the additive exponential
noise (AEN) channel, and lossy compression of exponential sources.
Taking advantage of the essential decomposable property of these channels (sources), the proposed expansion method allows for mapping of these problems to coding over parallel channels (respectively, sources), where each level is modeled as an independent coding problem over discrete alphabets. Any feasible solution to the resulting optimization problem after expansion corresponds to an achievable scheme of the original problem.
Utilizing this mapping, even for the cases where the optimal solutions are
difficult to characterize, it is shown that the expansion coding
scheme still performs well with appropriate choices of parameters.
More specifically, theoretical analysis and numerical results reveal that expansion coding achieves the capacity of AEN channel in the high SNR regime. It is also shown that for lossy compression, the achievable rate distortion pair by expansion coding approaches to the Shannon limit in the low distortion region.
Remarkably, by using capacity-achieving codes with low encoding and decoding complexity that are originally designed for discrete alphabets, for instance polar codes, the proposed
expansion coding scheme allows for designing low-complexity analog channel and source codes.
\end{abstract}

%%%%%%%%%%%%%%%%%%%%%%%%%%%%%%%%%%%%%%%%%%%%%%%%%%%%%%%%%%%%%%%%%%%%%%%%%%%%%%
%%%%%%%%%%%%%%%%%%%%%%%%%%%%%%%%%%%%%%%%%%%%%%%%%%%%%%%%%%%%%%%%%%%%%%%%%%%%%%

\section{Introduction}
\label{sec:Introduction}

\subsection{Background}

The field of channel coding was started with Shannon's famous theorem proposed in 1948 \cite{Shannon:IT48}, which shows that the channel capacity upper bounds the amount of information that can be reliably transmitted over a noisy communication channel. After this result, seeking for practical coding schemes that could approach channel capacity became a central objective for researchers. On the way from theory to practice, many coding schemes are proposed. Different types of codes emerge in improving the performance, giving consideration to the trade-off between coding complexity and error decay rate.

The history of channel coding traces back to the era of algebraic
coding, including the well-known Hamming codes \cite{Hamming:code50},
Golay codes \cite{Golay:code49}, Reed-Muller codes
\cite{Muller:code54}\cite{Reed:code54}, Reed-Solomon codes
\cite{Reed:RScode60}, lattice codes \cite{Conway:Lattice88}, and
others \cite{Williams:Coding1983}. However, despite enabling
significant advances in code design and construction, algebraic coding
did not turn out to be the most promising means to approach the
Shannon limit. The next era of probabilistic coding considered
approaches that involved optimizing code performance as a function of
coding complexity. This line of development included convolutional
codes \cite{Forney:convolutional70}, and concatenated codes
\cite{Forney:Contatenated66} at earlier times, as well as turbo codes
\cite{Berrou:Turbo93} and low-density parity-check (LDPC) codes
\cite{Mackay:LDPC95}\cite{Sipser:LDPC96} afterwards. Recently, polar codes \cite{Arikan:Channel08} have
been proved to achieve Shannon limit of binary-input symmetric
channels with low encoding and decoding complexity. In another recent
study \cite{Perry:Spinal11}\cite{Balakrishnan:Spinal12}, new types of
rateless codes, viz. spinal codes, are proposed to achieve the channel
capacity.

Another well-studied (and practically valuable) research direction in information theory is the problem of compression of continuous-valued sources. Given the increased importance of voice, video and other multimedia, all of which are typically ``analog'' in nature, the value associated with low-complexity algorithms to compress continuous-valued data is likely to remain significant in the years to come. For discrete-valued ``finite-alphabet'' problems, the associated coding theorem \cite{Cover:IT1991} and practically-meaningful coding schemes are well known. Trellis based quantizers \cite{Viterbi:Trellis74} are the first to achieve the rate distortion trade-off, but with encoding complexity scaling exponentially with the constraint length. Later, Matsunaga and Yamamoto \cite{Matsunaga:LDPC2003} show that  a low density parity check (LDPC) ensemble, under suitable conditions on the ensemble structure, can achieve the rate distortion bound using an optimal decoder. \cite{Wainwright:LDGM2010} shows that low density generator matrix (LDGM) codes, as the dual of LDPC codes, with suitably irregular degree distributions, empirically perform close to the Shannon rate-distortion bound with message-passing algorithms. More recently, polar codes are shown to be the first provably rate distortion limit achieving codes with low complexity \cite{Korada:Source10}. In the case of analog sources, although both practical coding schemes as well as theoretical analysis are very heavily studied, a very limited literature exists that connects the theory with low-complexity codes. The most relevant literature in this context is on lattice compression and its low-density constructions \cite{Zamir:Lattice09}. Yet, this literature is also limited in scope and application.

\subsection{Contributions and Organization}

The problem of coding over analog noise channels is highly non-trivial
in general. To this end, a method of modulation is commonly utilized
to map discrete inputs to analog signals for transmission through the
physical channel \cite{Costello:Channel07}. In this paper, we focus on
designing and coding over such mappings. In particular, we propose a
new coding scheme for general analog channels with moderate coding
complexity based on an expansion technique, where channel noise is
perfectly or approximately represented by a set of independent
discrete random variables (see
Fig.~\ref{fig:Expansion_Framework}). Via this representation, the
problem of coding over an analog noise channel is reduced to that of
coding over parallel discrete channels. We focus on additive
exponential noise (AEN), and we show that the Shannon limit, i.e., the
capacity, is achievable for AEN channel in the high SNR regime. More precisely, for any given $\epsilon<1$, it is shown that the gap to capacity is at most $5\epsilon \log(e)$ when at least $-2\log(e)+\log(\textrm{SNR})$ number of levels are utilized in the coding scheme together with embedded binary codes. Generalizing results to $q$-ary alphabets, we show that this gap can be reduced more. The main advantage of the proposed method lies on its complexity
inheritance property, where the encoding and decoding complexity of
the proposed schemes follow that of the embedded capacity achieving
codes designed for discrete channels, such as polar codes and spinal codes.

\begin{figure}[t]
 \centering
 \includegraphics[scale=1]{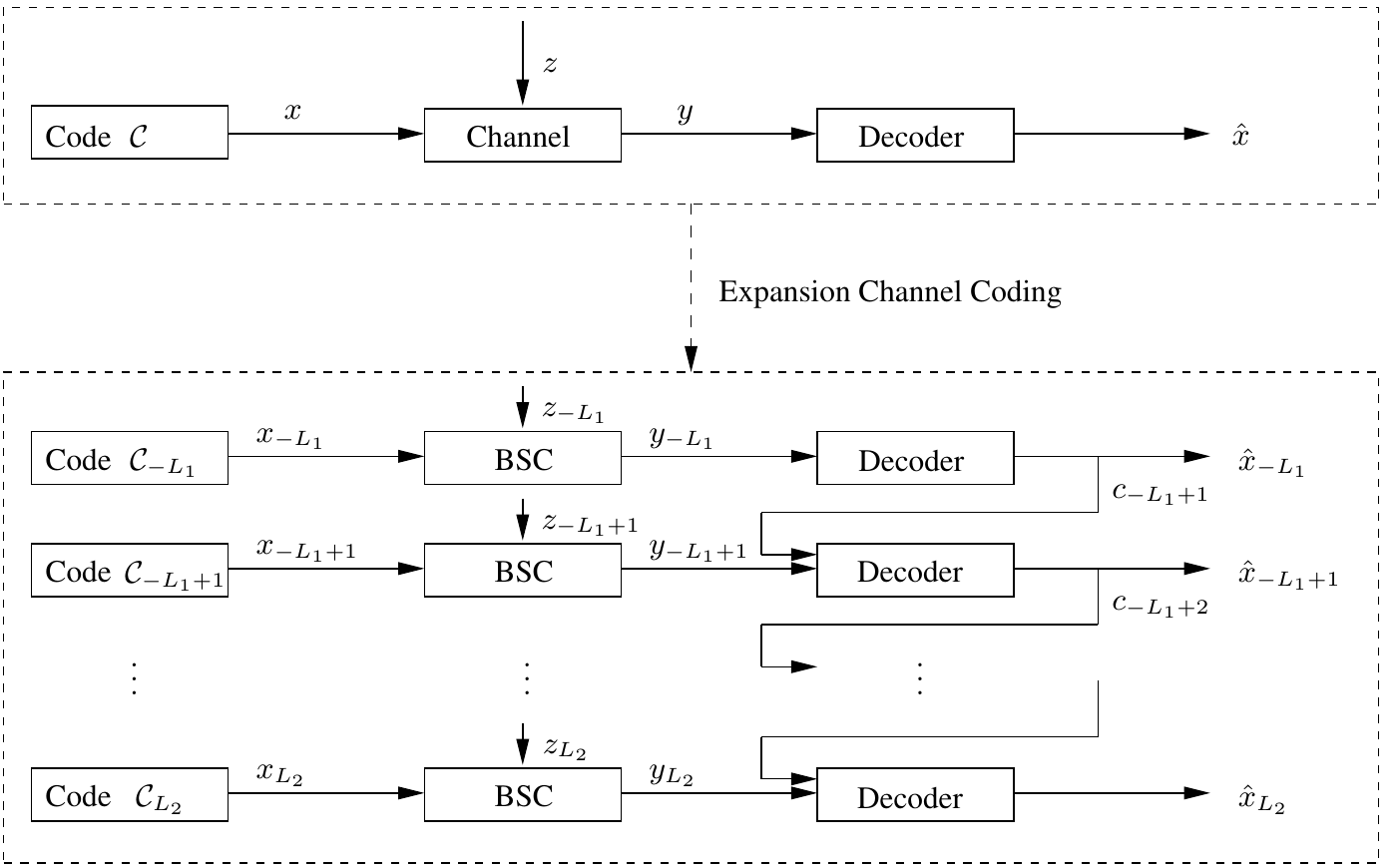}
 \caption{{\bf Illustration of expansion channel coding framework.} In this example, an analog noise channel is expanded into a set of discrete channels with index from $-L_1$ to $L_2$. Channel noise $z$ is considered as its binary expansion $z=\cdots z_{L_2}\cdots z_1z_0.z_{-1}\cdots z_{-L_1}\cdots$, and similar expansions are adopted to channel input and output. Carries exist between neighboring levels.
}
\label{fig:Expansion_Framework}
\end{figure}

In the second part of this paper, we present an expansion coding scheme for compressing of analog sources. This is a dual problem to the channel coding case, and we utilize a similar approach where we consider expanding exponential sources into binary sequences, and coding over the resulting set of parallel discrete sources. We show that this scheme's performance can get very close to the rate distortion limit in the low distortion regime (i.e., the regime of interest in practice). More precisely, the gap between the rate of the proposed scheme and the theoretical limit is shown to be within a constant gap ($5\log(e)$) for any distortion level $D$ when at least $-2\log(\lambda ^2 D)$ number of levels are utilized in the coding scheme (where, $1/\lambda$ is the mean of the underlying exponential source). Moreover, this expansion coding scheme can be generalized to Laplacian sources (two-sided symmetric exponential distribution), where the sign bit is considered separately and encoded perfectly to overcome the difficulty of source value being negative.

The rest of paper is organized as follows. Related work is
investigated and summarized in Section~\ref{sec:Related_Work}. The
expansion coding scheme for channel coding is detailed and evaluated in Section~\ref{sec:Channel_Coding}. The
expansion source coding framework and its application to exponential
sources are demonstrated in Section~\ref{sec:Source_Coding}. Finally,
the paper is concluded in Section~\ref{sec:Conclusion}.

%%%%%%%%%%%%%%%%%%%%%%%%%%%%%%%%%%%%%%%%%%%%%%%%%%%%%%%%%%%%%%%%%%%%%%%%%%%%%%
%%%%%%%%%%%%%%%%%%%%%%%%%%%%%%%%%%%%%%%%%%%%%%%%%%%%%%%%%%%%%%%%%%%%%%%%%%%%%%

\section{Related Work}
\label{sec:Related_Work}

Multilevel coding is a general coding method designed for analog noise
channels with a flavor of expansion \cite{Forney:Sphere00}. In
particular, a lattice partition chain
$\Lambda_1/\cdots/\Lambda_{r-1}/\Lambda_r$ is utilized to represent
the channel input, and, together with a shaping technique, the
reconstructed codeword is transmitted to the channel. It has been
shown that optimal lattices achieving Shannon limit exist. However, the
encoding and decoding complexity for such codes is high, in
general. In the sense of representing the channel input, our scheme is
coincident with multilevel coding by choosing
$\Lambda_1=q^{-L_2}\mathbb{Z}$, \ldots, $\Lambda_r=q^{L_1}\mathbb{Z}$,
for some $L_1,L_2\in\mathbb{Z}^+$, where coding of each level is over
$q$-ary finite field (see Fig.~\ref{fig:Multilevel_Framework}). The
difference in the proposed method is that besides representing the
channel input in this way, we also ``expand'' the channel noise, such
that the coding problem for each level is more suitable to solve by
adopting existing discrete coding schemes with moderate coding
complexity. Moreover, by adapting the underlying codes to channel-dependent variables, such as carries, the Shannon limit is shown to be
achievable by expansion coding with moderate number of expanded
levels.

\begin{figure}[t]
 \centering
 \includegraphics[scale=1]{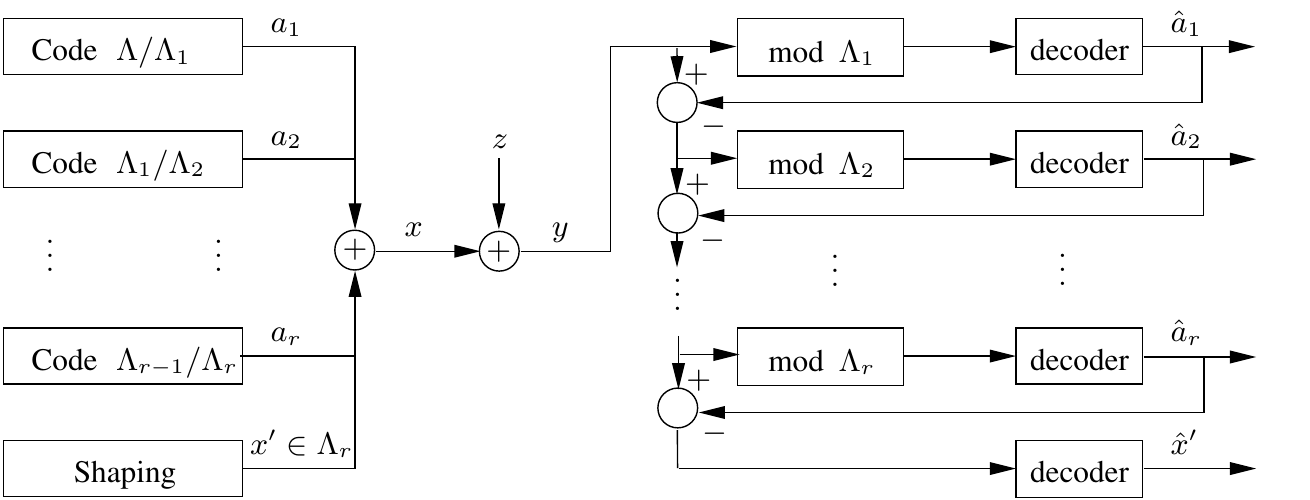}
 \caption{{\bf Illustration of multilevel coding framework.} In this example, multilevel coding scheme is illustrated. Compared to expansion coding in Fig.~\ref{fig:Expansion_Framework}, only channel input is expressed by multi-levels, but not the channel noise.
}
\label{fig:Multilevel_Framework}
\end{figure}

The deterministic model, proposed in \cite{Avestimehr:Wireless11}, is another framework to study analog noise channel coding problems, where the basic idea
is to construct an approximate channel for which the
transmitted signals are assumed to be noiseless above a certain noise level. This approach has proved to be very effective in analyzing the capacity of networks. In particular, it has been shown that this framework perfectly characterizes degrees of freedom of point-to-point AWGN channels, as well as some multi-user channels. In this sense, our expansion coding scheme can be seen as a generalization
of these deterministic approaches. Here, the effective noise
in the channel is carefully calculated and the system takes
advantage of coding over the noisy levels at any SNR. This
generalized channel approximation approach can be useful in
reducing the large gaps reported in the previous works, because the noise
approximation in our work is much closer to the actual distribution as compared
to that of the deterministic model (see Fig.~\ref{fig:Expansion_Deterministic_Compare}).

\begin{figure}[t]
 \centering
 \includegraphics[scale=1]{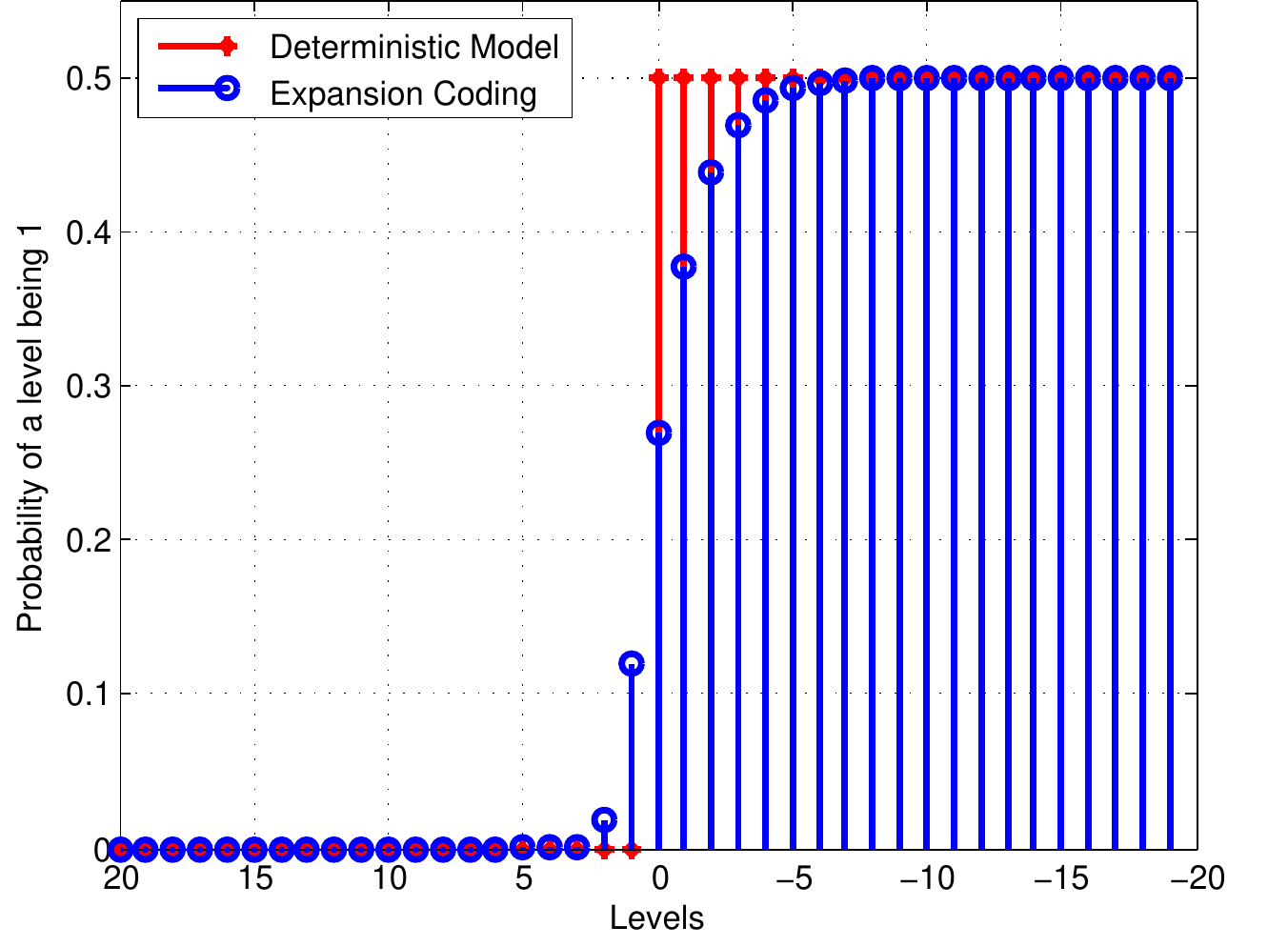}
 \caption{{\bf Comparison of noise models between expansion coding and deterministic model.} The noise models of each level for expansion coding and deterministic model are illustrated. Deterministic model cut the noise to a certain level, and expansion coding has a smooth transaction regime.
}
\label{fig:Expansion_Deterministic_Compare}
\end{figure}

There have been many attempts to utilize discrete codes for analog
channels (beyond simple modulation methods). For example, after the
introduction of polar codes, considerable attention has been directed
towards utilizing their low complexity property for analog channel
coding. A very straightforward approach is to use the central limit
theorem, which says that certain combinations of i.i.d. discrete
random variables converge to a Gaussian distribution. As reported in
\cite{Balakrishnan:Spinal12} and \cite{Abbe:Polar11}, the capacity of
AWGN channel can be achieved by coding over large number of BSCs,
however, the convergence rate is linear which limits its application
in practice. To this end, \cite{Abbe:Polar12} proposes a MAC based
scheme to improve the convergence rate to exponential, at the expense of having a
much larger field size. A newly published result in
\cite{Seidl:Polar13} attempts to combine polar codes with multilevel
coding, however many aspects of this optimization of polar-coded
modulation still remain open.  Along the direction of this research,
we also try to utilize capacity achieving discrete codes to
approximately achieve the capacity of analog channels.

The additive exponential noise (AEN) channel is of particular interest as it
models worst-case noise given a mean and a non-negativity constraint on
noise \cite{Verdu:Exponential96}. In addition, the AEN model naturally
arises in non-coherent communication settings, and in optical communication
scenarios. (We refer to \cite{Verdu:Exponential96} and
\cite{Martinez:Communication11} for an extensive discussion on the
AEN channel.) Verd{\'u} derived the optimal input distribution and the capacity
of the AEN channel in \cite{Verdu:Exponential96}. Martinez, on the other
hand, proposed the pulse energy modulation scheme, which can be seen as a
generalization of amplitude modulation for the Gaussian channels.
In this scheme, the constellation symbols are chosen as $c (i-1)^l$,
for $i=1, \cdots, 2^M$ with a constant $c$, and it is shown that
the information rates obtained from this constellation can achieve an
energy (SNR) loss of $0.76$ dB (with the best choice of
$l=\frac{1}{2}(1+\sqrt{5})$) compared to the capacity in the high SNR
regime. Another constellation technique for this coded modulation
approach is recently considered in \cite{LeGoff:Capacity11}, where
log constellations are designed such that
the real line is divided into ($2M-1$) equally probable intervals.
$M$ of the centroids of these intervals are chosen as constellation
points, and, by a numerical computation of the mutual information, it
is shown that these constellations can achieve within a $0.12$ dB SNR gap
in the high SNR regime. Our approach, which achieves arbitrarily close to
the capacity of the channel, outperforms these previously proposed
modulation techniques.

In the domains of image compression and speech coding, Laplacian and exponential distributions are widely adopted as natural models of correlation between pixels and amplitude of voice \cite{Gallager:Information68}. Exponential distribution is also fundamental in characterizing continuous-time Markov processes \cite{Verdu:Exponential96}. Although the rate distortion functions for both have  been known for decades, there is still a gap between theory and existing low-complexity coding schemes. The proposed schemes, primarily for the medium to high distortion regime, include the classical scalar and vector quantization schemes \cite{gray1998quantization}, and Markov chain Monte Carlo (MCMC) based approach in \cite{Baron:MCMC12}. However, the understanding of low-complexity coding schemes, especially for the low-distortion regime, remains limited. To this end, our expansion source coding scheme aims to approach the rate distortion limit with practical encoding and decoding complexity. By expanding the sources into independent levels, and using the decomposition property of exponential distribution, the problem has been remarkably reduced to a set of simpler subproblems, compression for discrete sources.

%%%%%%%%%%%%%%%%%%%%%%%%%%%%%%%%%%%%%%%%%%%%%%%%%%%%%%%%%%%%%%%%%%%%%%%%%%%%%%
%%%%%%%%%%%%%%%%%%%%%%%%%%%%%%%%%%%%%%%%%%%%%%%%%%%%%%%%%%%%%%%%%%%%%%%%%%%%%%

\section{Expansion Channel Coding}
\label{sec:Channel_Coding}

%%%%%%%%%%%%%%%%%%%%%%%%%%%%%%%%%%%%%%%%%%%%%%%%%%%%%%%%%%%%%%%%%%%%%%%%%%%%%%

\subsection{Intuition}

In general, expansion channel coding is a scheme of reducing the problem of coding over an
analog channel to coding over a set of discrete channels. In particular, we
consider the additive noise channel given by
\begin{equation}
\mathsf{Y}_i=\mathsf{X}_i+\mathsf{Z}_i,\; i=1,\cdots,n,\label{equ:General_Channel}
\end{equation}
where $\mathsf{X}_i$ are channel inputs with alphabet $\mathcal{X}$ (possibly having channel input requirements, such as certain moment constraints);
$\mathsf{Y}_i$ are channel outputs; $\mathsf{Z}_i$ are additive noises
independently and identically distributed with a continuous probability density function;
$n$ is block length. We represent the inputs as $\mathsf{X}_{1:n}\triangleq \{\mathsf{X}_1, \cdots,\mathsf{X}_n\}$. (Similar notation is used for other variables throughout the sequel.)

When communicating,
the transmitter conveys one of the messages, $\mathsf{M}$, which is
uniformly distributed in $\mathcal{M}\triangleq \{1, \cdots, 2^{nR}\}$; and it does so by
mapping the message to the channel input using encoding
function $\phi(\cdot):\mathcal{M} \to \mathcal{X}^n$ such that $\mathsf{X}_{1:n}(\mathsf{M})=\phi(\mathsf{M})$.
The decoder uses the decoding function $\psi(\cdot)$ to map its channel
observations to an estimate of the message. Specifically,
$\psi(\cdot):\mathcal{Y}^n \to \mathcal{M}$, where the estimate is denoted
by $\hat{\mathsf{M}} \triangleq \psi(\mathsf{Y}_{1:n})$. A rate $R$ is said to be achievable,
if the average probability of error defined by
\begin{equation}
P_{\textrm{e}} \triangleq \frac{1}{|\mathcal{M}|} \sum\limits_{\mathsf{M}\in \mathcal{M}}
\mathrm{Pr}\{\hat{\mathsf{M}}\neq \mathsf{M} |\; \mathsf{M} \textrm{ is sent.}\}\nonumber
\end{equation}
can be made arbitrarily small for large $n$.
The capacity of this channel is denoted by $C$, which is
the maximum achievable rate $R$, and its corresponding optimal input distribution is denoted as $f_{\mathsf{X}}^*(x)$.

Our proposed coding scheme is based on the idea that by ``expanding'' the channel noise (i.e., representing it by its $q$-ary expansion),
an approximate channel can be constructed, and proper coding schemes can be adopted to each level in this representation. If the approximation is close enough, then the coding schemes that are optimal for each level can be translated to an effective one for the original channel.
More formally, consider the original noise $\mathsf{Z}$ and its approximation $\hat{\mathsf{Z}}$, which is defined by the truncated $q$-ary expansion of $\mathsf{Z}$. For this moment, we simply take $q=2$ (i.e., considering binary expansion), and leave the general case for later discussion.
\begin{equation}
\hat{\mathsf{Z}}\triangleq \mathsf{Z}^{\textrm{sign}}\sum_{l=-L_1}^{L_2} 2^l \mathsf{Z}_l,\nonumber
\end{equation}
where $\mathsf{Z}^{\textrm{sign}}$ represents the sign of $\mathsf{Z}$, taking a value from $\{-,+\}$; $\mathsf{Z}_l$'s are mutually independent Bernoulli random variables. By similarly expanding the channel input, we convert the problem of coding over analog channels to coding over a set of binary discrete channels. This mapping is highly advantageous, as capacity achieving discrete codes can be adopted for coding over the constructed binary channels. Assume the input distributions for sign channel and discrete channel at $l$ are represented by $\mathsf{X}^{\textrm{sign}}$ and $\mathsf{X}_l$ correspondingly, then an achievable rate (via random coding) for the approximated channel is given by
\begin{equation}
\hat{R}\triangleq I(\hat{\mathsf{X}};\hat{\mathsf{X}}+\hat{\mathsf{Z}}),\nonumber
\end{equation}
where
\begin{equation}
\hat{\mathsf{X}}\triangleq \mathsf{X}^{\textrm{sign}}\sum\limits_{l=-L_1}^{L_2}2^l \mathsf{X}_l.\nonumber
\end{equation}
By adopting the same coding scheme over the original channel, one can achieve a rate given by
\begin{equation}
R\triangleq I(\hat{\mathsf{X}};\hat{\mathsf{X}}+\mathsf{Z}).\nonumber
\end{equation}
The following result provides a theoretical basis for expansion coding. (Here, $\overset{d.}{\to}$ denotes convergence in distribution.)
\begin{theorem}\label{thm:Channel_Expansion_Coding}
If $\hat{\mathsf{Z}}\overset{d.}{\to}\mathsf{Z}$ and $\hat{\mathsf{X}}\overset{d.}{\to}\mathsf{X}^*$, as $L_1,L_2\to\infty$, where $\mathsf{X}^*\sim f^*_{\mathsf{X}}(x)$, i.e., the optimal input distribution for the original channel, then $R\to C$.
\end{theorem}
The proof of this theorem follows from the continuity property of
mutual information. In other words, if the approximate channel is
close to the original one, and the distribution of the input is
close to the optimal input distribution, then the expansion coding
scheme will achieve the capacity of the channel under consideration.

%%%%%%%%%%%%%%%%%%%%%%%%%%%%%%%%%%%%%%%%%%%%%%%%%%%%%%%%%%%%%%%%%%%%%%%%%%%%%%

\subsection{AEN Channel Coding Problem Setup}

In this section, we consider an example where expansion channel coding can
achieve the capacity of the target channel. The particular
channel considered is an additive exponential noise (AEN) channel, where the channel noise $\mathsf{Z}_i$ in \eqref{equ:General_Channel} is independently and identically distributed according to an exponential density with mean $E_{\mathsf{Z}}$, i.e., omitting the index $i$, noise has the following density:
\begin{equation}
f_{\mathsf{Z}}(z) = \frac{1}{E_{\mathsf{Z}}} e^{-\frac{z}{E_{\mathsf{Z}}}} u(z),\label{equ:AEN_Noise_Distribution}
\end{equation}
where $u(z)=1$ for $z\geq 0$ and $u(z)=0$ otherwise. Moreover, channel input $\mathsf{X}_i$ in \eqref{equ:General_Channel} is restricted to be non-negative and satisfies the mean constraint
\begin{equation}
\frac{1}{n}\sum_{i=1}^n\mathbb{E}[\mathsf{X}_i]\leq E_{\mathsf{X}}.\label{equ:AEN_Input_Constraint}
\end{equation}

The capacity of AEN channel is given by \cite{Verdu:Exponential96},
\begin{equation}
C_{\textrm{AEN}} = \log(1+\textrm{SNR}),\label{equ:AEN_Capacity}
\end{equation}
where $\textrm{SNR}\triangleq E_{\mathsf{X}}/E_{\mathsf{Z}}$, and the capacity
achieving input distribution is given by
\begin{equation}
f^*_{\mathsf{X}}(x) = \frac{E_{\mathsf{X}}}{(E_{\mathsf{X}}+E_{\mathsf{Z}})^2} e^{\frac{-x}{E_{\mathsf{X}}+E_{\mathsf{Z}}}}u(x)
+ \frac{E_{\mathsf{Z}}}{E_{\mathsf{X}}+E_{\mathsf{Z}}} \delta(x),\label{equ:AEN_Optimal_Input}
\end{equation}
where $\delta(x)=1$ if and only if $x=0$.
Here, the optimal input distribution is not exponentially distributed, but a mixture of an exponential distribution with a delta function.
However, we observe that in the high SNR regime, the optimal
distribution gets closer to an exponential distribution with mean
$E_{\mathsf{X}}$, since the weight of delta function approaches to $0$ as SNR tends to infinity.

%%%%%%%%%%%%%%%%%%%%%%%%%%%%%%%%%%%%%%%%%%%%%%%%%%%%%%%%%%%%%%%%%%%%%%%%%%%%%%

\subsection{Binary Expansion of Exponential Distribution}

The basis of the proposed coding scheme is the expansion of analog random variables to discrete ones, and the exponential distribution emerges as a first candidate due to its decomposition property.
We show the following lemma, which allows us to have independent
Bernoulli random variables in the binary expansion of an exponential
random variable.
\begin{lemma}
\label{lem:Exponential_Expansion}
Let $\mathsf{B}_l$'s be independent Bernoulli random variables with parameters
given by $b_l$, i.e., $\textrm{Pr}\{\mathsf{B}_l=1\}\triangleq b_l$,
and consider the random variable defined by
\begin{equation}
\mathsf{B}\triangleq \sum\limits_{l=-\infty}^{\infty} 2^l \mathsf{B}_l.\nonumber
\end{equation}
Then, the random variable $\mathsf{B}$ is exponentially distributed
with mean $\lambda^{-1}$, i.e., its pdf is given by
\begin{equation}
f_{\mathsf{B}}(b) = \lambda e^{-\lambda b}, \quad b\geq 0,\nonumber
\end{equation}
if and only if the choice of $b_l$ is given by
\begin{equation}
b_l = \frac{1}{1+e^{\lambda 2^l}}. \nonumber
\end{equation}
\end{lemma}
\begin{proof}
See Appendix~\ref{app:Exponential_Expansion_Proof}.
\end{proof}
This lemma reveals that one can reconstruct exponential random variable from a set of independent Bernoulli random variables perfectly. Fig.~\ref{fig:Exponential_Recovery} illustrates that the distribution of recovered random variable from expanded levels (obtained from the statistics of $100,000$ independent samples) is a good approximation of original exponential distribution.
\begin{figure}[t]
 \centering
 \includegraphics[scale=1]{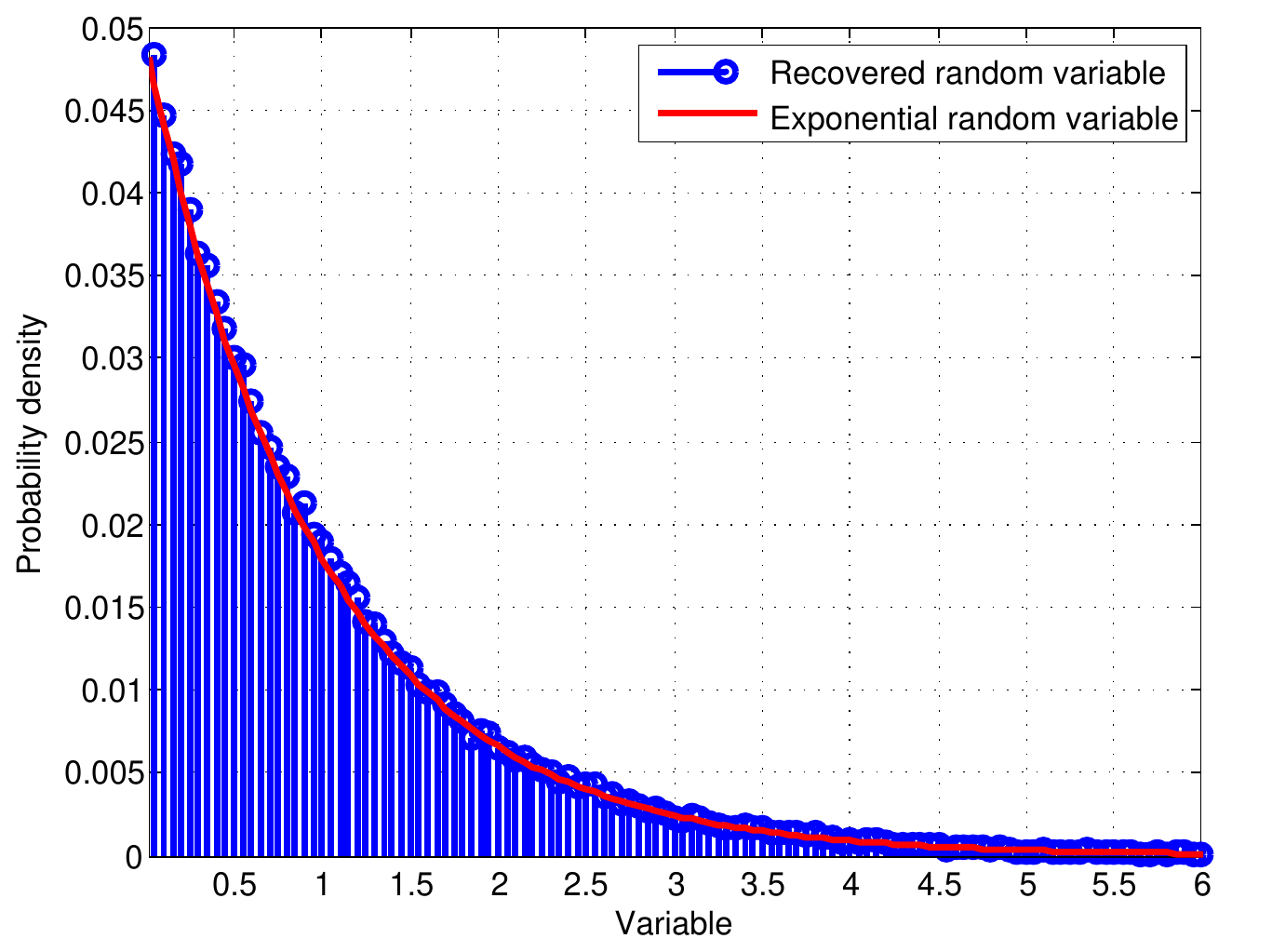}
 \caption{{\bf Distribution of recovered random variable from expanded levels, comparing with original exponential distribution ($\lambda=1$).} $100,000$ samples are generated from the expansion form of discrete random variables, where expansion levels are truncated from $-10$ to $10$.}
\label{fig:Exponential_Recovery}
\end{figure}

A set of typical numerical values of $b_l$s for $\lambda=1$ is shown in Fig.~\ref{fig:Exponential_Expansion_Parameter}. It is evident that $b_l$ approaches to $0$
for the ``higher'' levels and approaches $0.5$ for what we refer to as  ``lower'' levels. Hence, the primary non-trivial levels for which coding is meaningful are the so-called ``middle'' ones, which provides the basis for  truncating the number of levels to a finite value without a significant loss in performance.
\begin{figure}[t]
 \centering
 \includegraphics[scale=1]{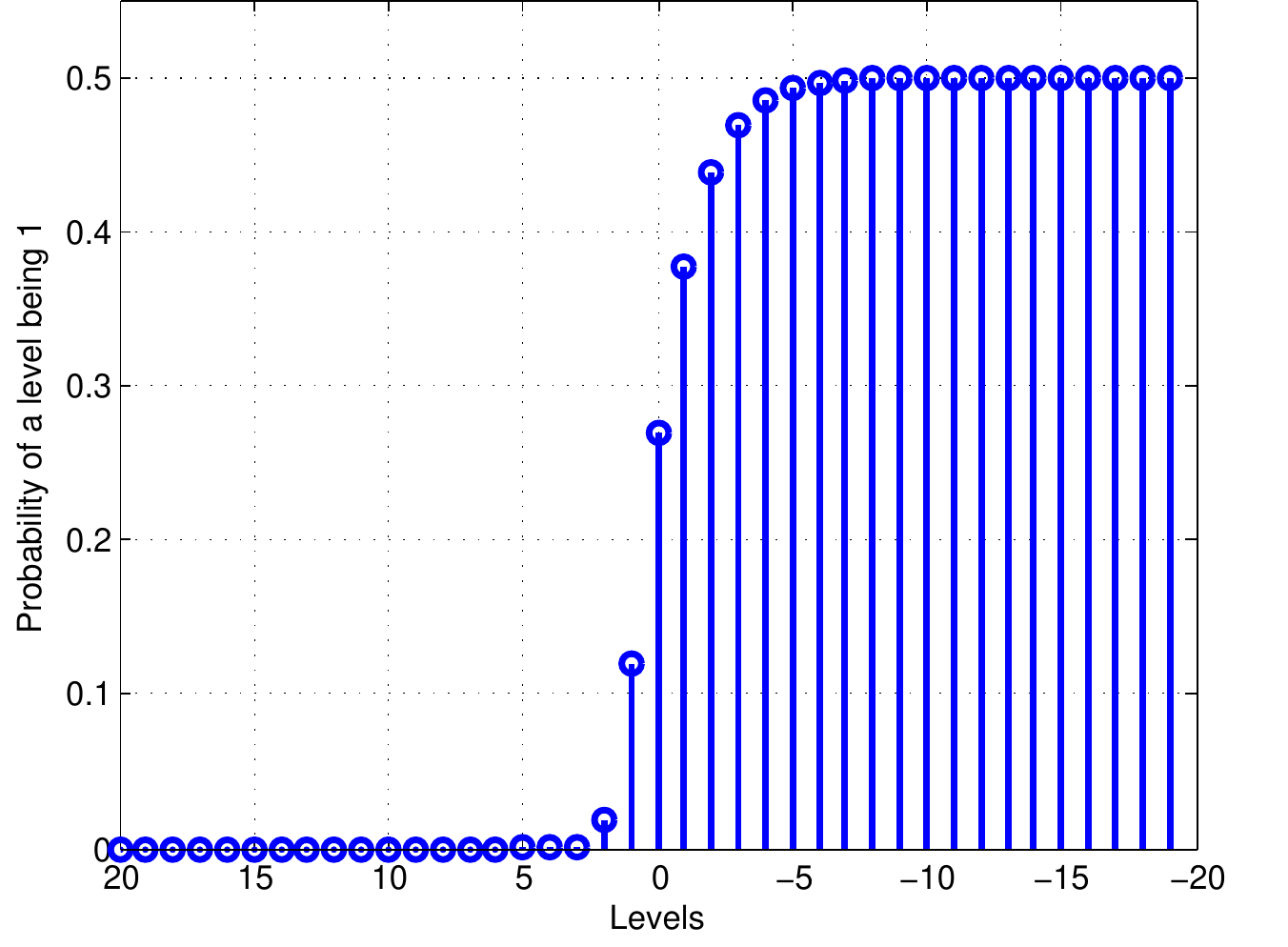}
\caption{{\bf Numerical results for a set of $b_l$ with $\lambda=1$.} X-axis is the level index for binary expansion (e.g., value $-2$ means the weight of corresponding level is $2^{-2}$), and Y-axis
shows the corresponding probability of taking value $1$ at each level, i.e., $b_l$.
}
\label{fig:Exponential_Expansion_Parameter}
\end{figure}

%%%%%%%%%%%%%%%%%%%%%%%%%%%%%%%%%%%%%%%%%%%%%%%%%%%%%%%%%%%%%%%%%%%%%%%%%%%%%%

\subsection{Expansion Coding for AEN Channel}

We consider the binary expansion of the channel noise
\begin{align}
\hat{\mathsf{Z}}_i \triangleq \sum\limits_{l=-L_1}^{L_2} 2^{l} \mathsf{Z}_{i,l},\label{equ:AEN_Noise_Expansion}
\end{align}
where $\mathsf{Z}_{i,l}$ are i.i.d. Bernoulli random variables with parameters \begin{equation}
q_l\triangleq\textrm{Pr}\{\mathsf{Z}_{l}=1\}=\frac{1}{1+e^{ 2^l/E_{\mathsf{Z}}}},\quad l=-L_1,\ldots,L_2.\label{equ:AEN_Noise_Parameter}
\end{equation}
By Lemma \ref{lem:Exponential_Expansion}, $\hat{\mathsf{Z}}_i\overset{d.}{\to} \mathsf{Z}_i$ as $L_1, L_2\to \infty$. In this sense, we approximate the exponentially distributed noise perfectly by a set of discrete Bernoulli distributed noises. Similarly, we also expand channel input and output as in the following,
\begin{align}
&\hat{\mathsf{X}}_i \triangleq \sum\limits_{l=-L_1}^{L_2} 2^{l} \mathsf{X}_{i,l},\label{equ:AEN_Input_Expansion}\\
&\hat{\mathsf{Y}}_i \triangleq \sum\limits_{l=-L_1}^{L_2} 2^{l} \mathsf{Y}_{i,l},\label{equ:AEN_Output_Expansion}
\end{align}
where $\mathsf{X}_{i,l}$ and $\mathsf{Y}_{i,l}$ are also Bernoulli random variables with parameters $\textrm{Pr}\{\mathsf{X}_{l}=1\}\triangleq p_l$ and $\textrm{Pr}\{\mathsf{Y}_{l}=1\}\triangleq r_l$ correspondingly. Here, the channel input is chosen as zero for levels $l\notin\{-L_1, \cdots, L_2\}$.
Noting that the summation in the original channel is a sum over real numbers,
we do not have a binary symmetry channel (BSC) at each level (from $\mathsf{X}_l$s to $\mathsf{Y}_l$s).
If we could replace the real sum by modulo-$2$ sum such that at each level $l$ we have an independent coding problem, then any capacity achieving BSC code can be utilized over this channel. (Here, instead of directly using the capacity achieving
input distribution of each level, we can use its combination with the method of Gallager~\cite{Gallager:Information68} to
achieve a rate corresponding to the one obtained by the mutual
information $I(\mathsf{X}_l;\mathsf{Y}_l)$ evaluated with an input distribution
Bernoulli with parameter $p_l$. This helps to approximate the optimal input distribution of the original channel.)
However, due to the addition over real numbers, carries exist between
neighboring levels, which further implies that the levels are not independent.
Every level, except for the lowest one, is impacted by carry from lower levels. In order to alleviate this issue, two schemes are proposed in the following to ensure independent operation of the levels. In these models of coding over independent parallel channels, the total achievable rate is the summation of individual achievable rates over all levels.

\subsubsection{Considering carries as noise}

Denoting the carry seen at level $l$ as
$\mathsf{C}_{i,l}$, which is also a Bernoulli random variable with parameter $\textrm{Pr}\{\mathsf{C}_{i,l}=1\}\triangleq c_l$, the remaining channels can be represented
with the following,
\begin{equation}
\mathsf{Y}_{i,l} = \mathsf{X}_{i,l} \oplus \tilde{\mathsf{Z}}_{i,l}, \quad i=1,\cdots n,\nonumber
\end{equation}
where the effective noise, $\tilde{\mathsf{Z}}_{i,l}$,
is a Bernoulli random variable obtained
by the convolution of the actual noise and the carry, i.e.,
\begin{equation}
\tilde{q}_{l} \triangleq \textrm{Pr}\{ \tilde{\mathsf{Z}}_{i,l}=1\}
= q_{l} \otimes c_{l}\triangleq q_l(1-c_l)+c_l(1-q_l).\nonumber
\end{equation}
Here, the carry probability is given by the following recursion relationship:
\begin{itemize}
\item For level $l=-L_1$,
\begin{equation}
c_{-L_1} = 0;\nonumber
\end{equation}
\item For level $l>-L_1$,
\begin{align}
c_{l+1}=p_{l}q_{l}(1-c_{l})+p_{l}(1-q_{l})c_{l}+(1-p_{l})q_{l}c_l+p_lq_lc_l.\nonumber
\end{align}
\end{itemize}

Using capacity achieving codes for BSC, e.g., polar codes or spinal codes,
combined with the Gallager's method, expansion coding achieves the following rate by considering carries as noise.
\begin{theorem}
\label{thm:AEN_Expansion_Coding_SchemeI}
Expansion coding, considering carries as noise, achieves the rate for AEN channel given by
\begin{equation}
\hat{R}_1=\sum\limits_{l=-L_1}^{L_2}\hat{R}_{1,l}= \sum\limits_{l=-L_1}^{L_2} \left[H(p_l\otimes \tilde{q}_l) - H(\tilde{q}_l)\right],\label{equ:AEN_Achievable_Rate_SchemeI}
\end{equation}
for any $L_1,L_2>0$, where $p_l\in[0,0.5]$
is chosen to satisfy constraint \eqref{equ:AEN_Input_Constraint}, i.e.,
\begin{equation}
\frac{1}{n} \sum\limits_{i=1}^n \mathbb{E}[\hat{\mathsf{X}}_i ]=\frac{1}{n} \sum\limits_{i=1}^n\sum\limits_{l=-L_1}^{L_2}2^l \mathbb{E}[\mathsf{X}_{i,l} ]
 = \sum\limits_{l=-L_1}^{L_2} 2^{l} p_l
\leq E_{\mathsf{X}}.\nonumber
\end{equation}
\end{theorem}

\subsubsection{Decoding carries}

In this scheme, let us consider decoding starting from the lowest level
$l=-L_1$. The receiver will obtain the correct $\mathsf{X}_{i,-L_1}$ for
$i=1,\cdots,n$ by using powerful discrete coding at this level. As the receiver has the knowledge of $\mathsf{Y}_{i,-L_1}$, it can determine
the correct noise sequence $\mathsf{Z}_{i,-L_1}$ for $i=1,\cdots,n$.
With this knowledge, the receiver can directly obtain each
$\mathsf{C}_{i,-L_1+1}$ for $i=1,\cdots,n$, which is the carry from
level $l=-L_1$ to level $l=-L_1+1$. This way, by iterating to higher levels, the receiver can recursively subtract the impact of
carry bits. Therefore, when there is no decoding error at each level, the effective channel that
the receiver observes is given by
\begin{equation}
\mathsf{Y}_{i,l} = \mathsf{X}_{i,l} \oplus \mathsf{Z}_{i,l}, \quad i=1, \cdots, n,\nonumber
\end{equation}
for $l=-L_1,\cdots,L_2$.
We remark that with this decoding strategy, the effective channels
will no longer be a set of independent parallel channels,
as decoding in one level affects the channels at higher levels.
However, if the utilized coding method is strong enough (e.g.,
if the error probability decays to $0$ exponentially with $n$),
then decoding error due to carry bits can be made insignificant
by increasing $n$ for a given number of levels.
We state the rate resulting from this approach in the following theorem.

\begin{theorem}
\label{thm:AEN_Expansion_Coding_SchemeII}
Expansion coding, by decoding the carries, achieves
the rate for AEN channel given by
\begin{equation}
\hat{R}_2=\sum\limits_{l=-L_1}^{L_2}\hat{R}_{2,l}=\sum\limits_{l=-L_1}^{L_2} \left[H(p_l\otimes q_l) - H(q_l)\right],\label{equ:AEN_Achievable_Rate_SchemeII}
\end{equation}
for any $L_1,L_2>0$, where $p_l\in[0,0.5]$ is chosen to
satisfy constraint \eqref{equ:AEN_Input_Constraint}, i.e.,
\begin{equation}
\frac{1}{n} \sum\limits_{i=1}^n \mathbb{E}[\hat{\mathsf{X}}_i ]=\frac{1}{n} \sum\limits_{i=1}^n\sum\limits_{l=-L_1}^{L_2}2^l \mathbb{E}[\mathsf{X}_{i,l} ]
 = \sum\limits_{l=-L_1}^{L_2} 2^{l} p_l
\leq E_{\mathsf{X}}.\nonumber
\end{equation}
\end{theorem}

Compared to the previous case, the optimization problem is simpler
here as the rate expression is simply the sum of the rates obtained
from a set of parallel channels.
Optimizing for these two theoretical achievable rates require choosing proper values for
$p_l$. Note that, the optimization problems given by Theorem~ \ref{thm:AEN_Expansion_Coding_SchemeI} and \ref{thm:AEN_Expansion_Coding_SchemeII} are not easy to solve in general. Here, instead of searching for the  optimal solutions directly, we utilize the information from the optimal input distribution of the original channel. Recall that the distribution in
\eqref{equ:AEN_Optimal_Input} can be approximated by an exponential distribution with
mean $E_{\mathsf{X}}$ at high SNR. Hence, one can simply choose $p_l$ from the binary expansion of the exponential
distribution with mean $E_{\mathsf{X}}$ as an achievable scheme, i.e.,
\begin{equation}
p_l\triangleq\textrm{Pr}\{\mathsf{X}_{l}=1\}=\frac{1}{1+e^{2^l/E_{\mathsf{X}}}},\quad l=-L_1,\ldots,L_2.\label{equ:AEN_Input_Parameter}
\end{equation}

We now show that this proposed scheme achieves the
capacity of AEN channel in the high SNR regime for a sufficiently
high number of levels. For this purpose, we first characterize the asymptotic behavior of entropy at each level for $q_l$ and $\tilde{q}_l$ correspondingly, where the later one is closely related to carries.
\begin{lemma}
\label{lem:AEN_Entropy_Bound}
The entropy of noise seen at level $l$, $H(q_l)$, is bounded by
\begin{align}
&H(q_l)< 2^{-l+\eta}3\log e\quad\text{for }\;l> \eta,\label{equ:AEN_Entropy_Bound1}\\
&H(q_l)> 1-2^{l-\eta}\log e\quad\text{for }\;l\leq \eta,\label{equ:AEN_Entropy_Bound2}
\end{align}
where $\eta\triangleq \log E_{\mathsf{Z}}$.
\end{lemma}
\begin{proof}
See Appendix~\ref{app:AEN_Entropy_Bound}.
\end{proof}

\begin{lemma}
\label{lem:AEN_Equivalent_Entropy_Bound}
The entropy of equivalent noise at level $l$, $H(\tilde{q}_l)$,
is bounded by
\begin{align}
&H(\tilde{q}_l)< 6(l-\eta)2^{-l+\eta} \log e\quad\text{for }\;l> \eta,\label{equ:AEN_Equivalent_Entropy_Bound1}\\
&H(\tilde{q}_l)> 1-2^{l-\eta}\log e\quad\text{for }\;l\leq \eta,\label{equ:AEN_Equivalent_Entropy_Bound2}
\end{align}
where $\eta\triangleq\log E_{\mathsf{Z}}$.
\end{lemma}
\begin{proof}
See Appendix~\ref{app:AEN_Equivalent_Entropy_Bound}.
\end{proof}

The intuitions behind these lemmas are given by the example scenario in Fig.~\ref{fig:AEN_Expanded_Level}, which shows that the bounds on noise tails are both exponential. Now, we state the main result indicating the capacity gap of expansion coding scheme over AEN channel.
\begin{theorem}
\label{thm:AEN_Achivable_Rate_Main_Result}
For any positive constant $\epsilon<1$, if
\begin{itemize}
\item $L_1\geq -\log \epsilon-\log E_{\mathsf{Z}}$;
\item $L_2\geq -\log \epsilon+\log E_{\mathsf{X}}$;
\item \emph{$\textrm{SNR}\geq 1/\epsilon$}, where \emph{$\textrm{SNR}=E_{\mathsf{X}}/E_{\mathsf{Z}}$},
\end{itemize}
then, with the choice of $p_l$ as \eqref{equ:AEN_Input_Parameter},
\begin{enumerate}
\item considering carries as noise, the achievable rate given by \eqref{equ:AEN_Achievable_Rate_SchemeI} satisfies
\emph{
\begin{equation}
\hat{R}_{1}\geq C_{\textrm{AEN}}-c,\nonumber
\end{equation}}
where $c$ is a constant independent of \emph{$\textrm{SNR}$} and $\epsilon$;
\item decoding carries, the achievable rate given by \eqref{equ:AEN_Achievable_Rate_SchemeII} satisfies
\emph{
\begin{equation}
\hat{R}_{2}\geq C_{\textrm{AEN}}-5\epsilon\log e .\nonumber
\end{equation}}
\end{enumerate}
\end{theorem}
\begin{proof}
The proof of this theorem is based on the observation that the sequence of $p_l$ is a left-shifted version of $q_l$ at high SNR regime.
As limited by power constraint, the number of levels shifted is at most $\log(1+\textrm{SNR})$, which further implies the rate we gain
is roughly $\log(1+\textrm{SNR})$ as well, when carries are decoded. If considering carries as noise, then there is apparent gap between the two version of noises, which leads to a constant gap for achievable rate.
Fig.~\ref{fig:AEN_Expanded_Level} helps to illustrate key steps of the intuition, and a detailed proof with precise
calculations is given in Appendix~\ref{app:AEN_Achivable_Rate_Main_Result}.
\end{proof}

\begin{figure}[t!]
 \centering
 \includegraphics[scale=1]{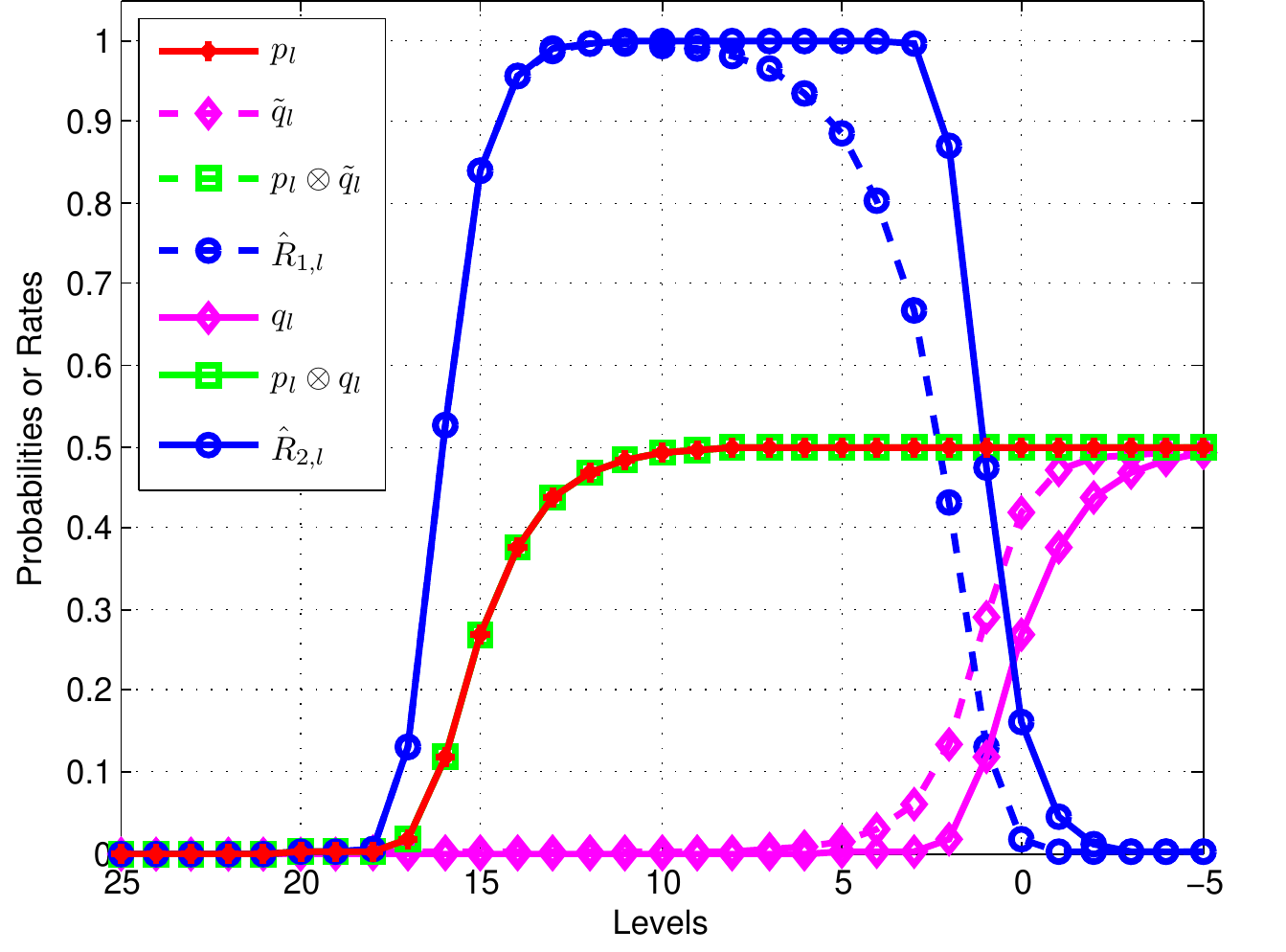}
 \caption{{\bf Signal and noise probabilities, and rates per level.}
$p_l$, $q_l$, $p_l\otimes q_l$, $\tilde{q}_l$, $p_l\otimes \tilde{q}_l$ and rates at each level are shown. In this example,
$E_{\mathsf{X}}=2^{15}$ and $E_{\mathsf{Z}}=2^0$, which further implies $p_l$ is a left-shifted version of $q_l$ by $15$ levels. The coding scheme with $L_1=5$ and $L_2=20$ covers the significant portion of the rate obtained by
using all of the parallel channels. }
\label{fig:AEN_Expanded_Level}
\end{figure}

By Lemma~\ref{lem:Exponential_Expansion},
$\hat{\mathsf{Z}}\overset{d.}{\to} \mathsf{Z}$, and combined with the
argument in Theorem~\ref{thm:Channel_Expansion_Coding}, we have
$\hat{R}_2\to R$ as $L_1,L_2\to\infty$. Hence, the coding scheme also
works well for the original AEN channel. More precisely, expansion coding scheme achieves the capacity of AEN channel at high SNR region using moderately large number of expansion levels.

%%%%%%%%%%%%%%%%%%%%%%%%%%%%%%%%%%%%%%%%%%%%%%%%%%%%%%%%%%%%%%%%%%%%%%%%%%%%%%
\subsection{Numerical results}

We calculate the rates obtained from the two schemes above
($\hat{R}_1$ as \eqref{equ:AEN_Achievable_Rate_SchemeI} and $\hat{R}_2$ as \eqref{equ:AEN_Achievable_Rate_SchemeII})
with input probability distribution given by (\ref{equ:AEN_Input_Parameter}).

Numerical results are given in Fig.~\ref{fig:AEN_Achievable_Rate}.
It is evident from the figure (and also from the analysis
given in Theorem~\ref{thm:AEN_Achivable_Rate_Main_Result}) that the proposed technique of decoding carries,
when implemented with sufficiently large number of levels,
achieves channel capacity at high SNR regime.

Another point is that neither of the two schemes works well in low SNR regime, which mainly results from the fact that
input approximation is only perfect for sufficiently high
SNR. Nevertheless, the scheme (the rate obtained by decoding carries) performs close to optimal in the moderate SNR regime as well.

\begin{figure}[t]
 \centering
 \includegraphics[scale=1]{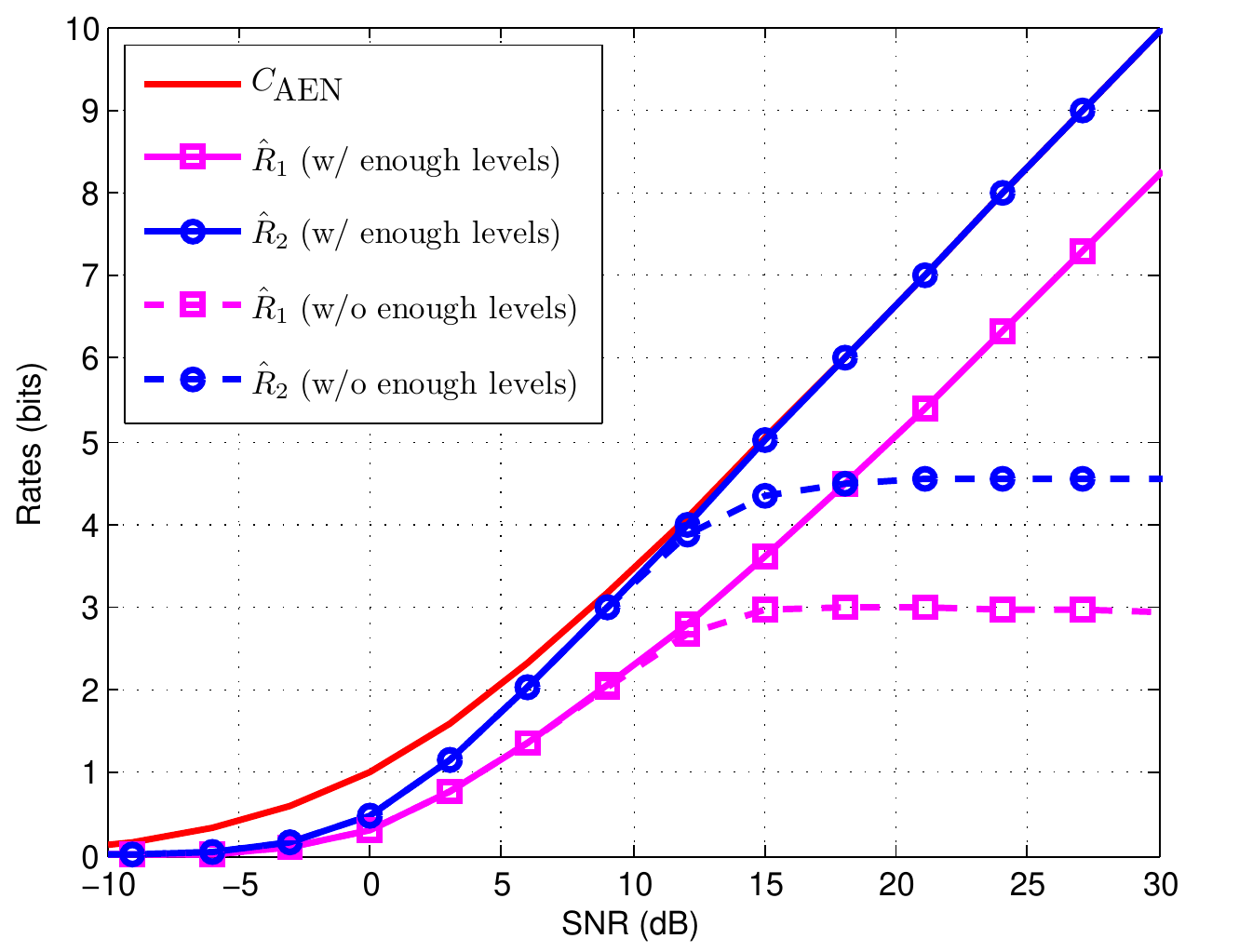}
 \caption{{\bf Numerical results of achievable rates for AEN channels using expansion coding.}
$\hat{R}_1$: The rate obtained by considering carries as noise.
$\hat{R}_2$: The rate obtained by decoding carry at each level.
Solid lines represent adopting enough number of levels as indicated in Theorem~\ref{thm:AEN_Achivable_Rate_Main_Result}, while dashed lines represent only adopting constant number of levels (not scaling with SNR).
}
\label{fig:AEN_Achievable_Rate}
\end{figure}

%%%%%%%%%%%%%%%%%%%%%%%%%%%%%%%%%%%%%%%%%%%%%%%%%%%%%%%%%%%%%%%%%%%%%%%%%%%%%%
\subsection{Generalization}
In the previous section, only binary expansion was considered. Generalization to $q$-ary expansion with $q\geq 2$ is discussed here. Note that this change does not impact the expansion coding framework, and the only difference lies in that each level after expansion should be modeled as a $q$-ary discrete memoryless channel. For this, we need to characterize the $q$-ary expansion of exponential distribution. Mathematically, the parameters of expanded levels for an exponential random variable $\mathsf{B}$ with parameter $\lambda$ can be calculated as follows:
\begin{align}
b_{l,s} &\triangleq \textrm{Pr}\{\mathsf{B}_l=s\}\nonumber\\
        &=\sum_{k=0}^{\infty}\textrm{Pr}\{q^l(qk+s)\leq\mathsf{B}<q^l(qk+s+1)\}\nonumber\\
        &=\sum_{k=0}^{\infty}\left[e^{-\lambda q^l(qk+s)}-e^{-\lambda q^l(qk+s+1)}\right]\nonumber\\
        &=\frac{\left(1-e^{-\lambda q^l}\right)e^{-\lambda q^l s}}{1-e^{-\lambda q^{l+1}}},\nonumber
\end{align}
where $l\in\{-L_1,\ldots,L_2\}$ and $s\in\{0,\ldots,q-1\}$.

Based on this result, consider channel input and noise expansions as
\begin{align}
p_{l,s} \triangleq \textrm{Pr}\{\mathsf{X}_l=s\}=\frac{\left(1-e^{- q^l/E_{\textsf{X}}}\right)e^{- q^l s/E_{\textsf{X}}}}{1-e^{-q^{l+1}/E_{\textsf{X}}}},\nonumber
\end{align}
and
\begin{align}
q_{l,s} \triangleq \textrm{Pr}\{\mathsf{Z}_l=s\}=\frac{\left(1-e^{- q^l/E_{\textsf{Z}}}\right)e^{- q^l s/E_{\textsf{Z}}}}{1-e^{-q^{l+1}/E_{\textsf{Z}}}}.\nonumber
\end{align}
Then, the achievable rate by decoding carries (note that in $q$-ary expansion case, carries are still Bernoulli distributed) can be expressed as
\begin{equation}
\hat{R}_2 = \sum_{l=-L_1}^{L_2}\left[ H(p_{l,0:q-1}\otimes q_{l, 0:q-1})-H(q_{l, 0:q-1})\right],\label{equ:AEN_Achivable_Rate_BaseQ}
\end{equation}
where $p_{l,0:q-1}$ and $q_{l, 0:q-1}$ denote the distribution of expanded random variables at level $l$ for input and noise respectively; $\otimes$ represents for the vector convolution.

When implemented with enough number of levels in coding, the achievable rates given by \eqref{equ:AEN_Achivable_Rate_BaseQ} achieves the capacity of AEN channel for any $q\geq 2$. More precisely, as shown in the numerical result in Fig.~\ref{fig:AEN_Achievable_Rate_BaseQ}, expansion coding with larger $q$ can achieve a higher rate (although this enhancement becomes limited when $q$ gets greater than $10$). This property of the coding scheme can be utilized to trade-off number of levels ($L_1+L_2$) and the alphabet size ($q$) to achieve a certain rate at a given $\mathrm{SNR}$.
\begin{figure}[t]
 \centering
 \includegraphics[scale=1]{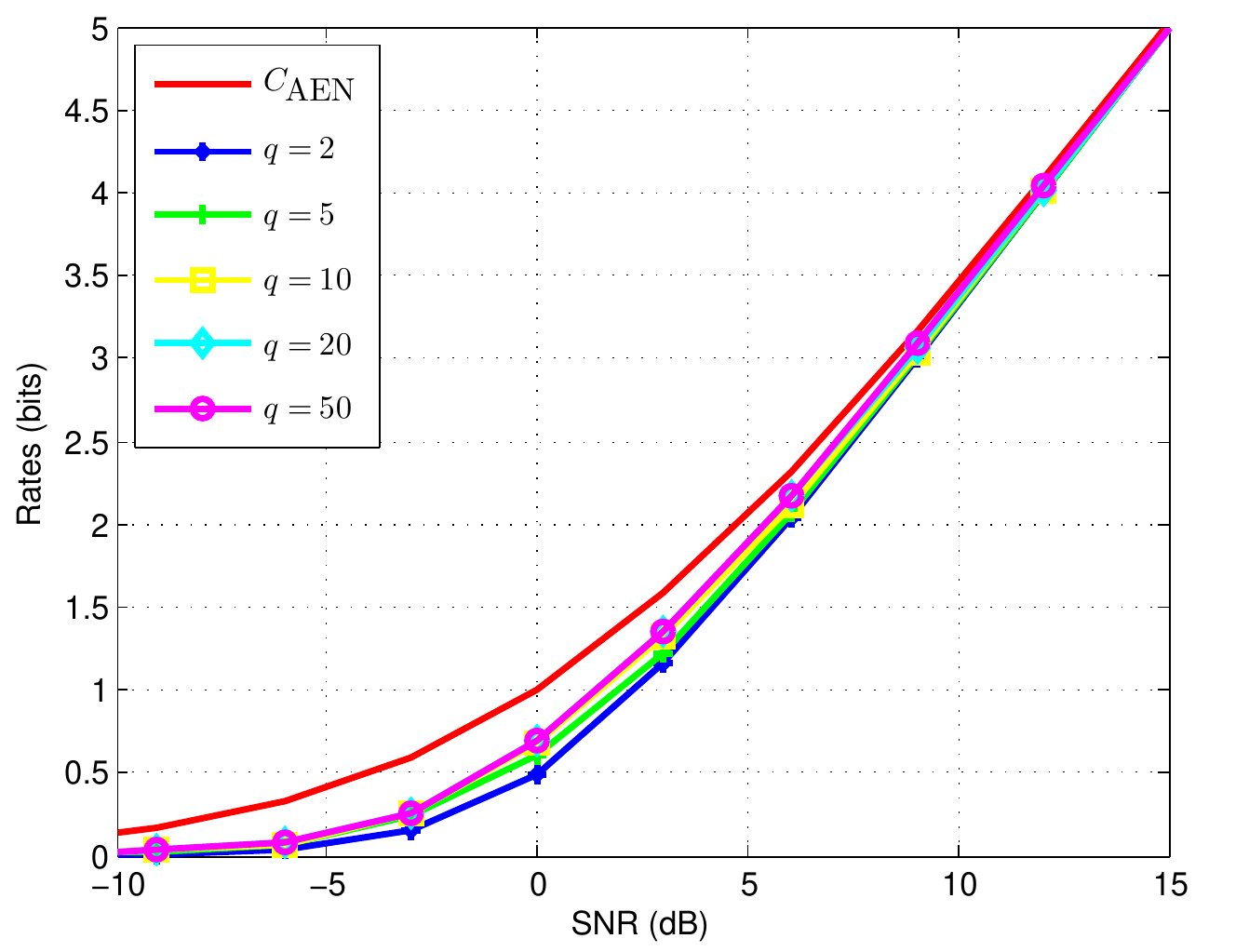}
 \caption{{\bf Numerical results for $q$-ary expansion.}
The achievable rates using $q$-ary expansion coding by decoding carries are illustrated.
}
\label{fig:AEN_Achievable_Rate_BaseQ}
\end{figure}
%%%%%%%%%%%%%%%%%%%%%%%%%%%%%%%%%%%%%%%%%%%%%%%%%%%%%%%%%%%%%%%%%%%%%%%%%%%%%%
%%%%%%%%%%%%%%%%%%%%%%%%%%%%%%%%%%%%%%%%%%%%%%%%%%%%%%%%%%%%%%%%%%%%%%%%%%%%%%

\section{Expansion Source Coding}
\label{sec:Source_Coding}

\subsection{Intuition}

Expansion source coding is a scheme of reducing the problem of compressing analog sources to compressing a set of discrete sources. In particular, consider an i.i.d. source $\mathsf{X}_1, \mathsf{X}_2,\ldots, \mathsf{X}_n$. A $(2^{nR},n)$-rate distortion code consists of an encoding function $\varphi:\mathbb{R}^n\to\mathcal{M}$, where $\mathcal{M}\triangleq\{1,\ldots,2^{nR}\}$, and a decoding function $\varsigma:\mathcal{M}\to\mathbb{R}^n$, which together map $\mathsf{X}_{1:n}$ to an estimate $\tilde{\mathsf{X}}_{1:n}$. Then, the rate and distortion pair $(R,D)$ is said to be achievable if there exists a sequence of $(2^{nR},n)$-rate distortion codes with $\lim\limits_{n\to\infty}\mathbb{E}[d(\mathsf{X}_{1:n},\tilde{\mathsf{X}}_{1:n})]\leq D$ for a given distortion measure of interest $d(\cdot,\cdot)$. The rate distortion function $R(D)$ is the infimum of such rates, and by Shannon's theorem \cite{Cover:IT1991}, we have:
\begin{equation}
R(D)=\min_{f(\tilde{x}|x):\mathbb{E}[d(\mathsf{X}_{1:n},\tilde{\mathsf{X}}_{1:n})]\leq D}I(\mathsf{X};\tilde{\mathsf{X}}),\nonumber
\end{equation}
where the optimal conditional distribution is given by $f^*(\tilde{x}|x)$.

The expansion source coding scheme proposed here is based on the observation that by expanding the original analog source into a set of independent discrete random variables, proper source coding schemes could be adopted for every expanded level. If this approximation in expansion is close enough, then the overall distortion obtained from expansion coding scheme is also close to the optimal distortion. More formally, consider the original analog source $\mathsf{X}$ and its approximation $\hat{\mathsf{X}}$ given by (omitting index $i$)
\begin{equation}
\hat{\mathsf{X}}\triangleq \mathsf{X}^{\text{sign}}\sum_{l=-L_1}^{L_2}2^l \mathsf{X}_l,
\end{equation}
where $\mathsf{X}^{\textrm{sign}}$ represents the sign of $\hat{\mathsf{X}}$ and takes values from $\{-,+\}$, and $\mathsf{X}_l$ is the expanded Bernoulli random variable at level $l$. Similarly, if we expand the estimate by
\begin{equation}
\hat{\tilde{\mathsf{X}}}\triangleq \tilde{\mathsf{X}}^{\textrm{sign}}\sum_{l=-L_1}^{L_2}2^l \tilde{\mathsf{X}}_l,
\end{equation}
where $\tilde{\mathsf{X}}^{\textrm{sign}}$ represents the sign of $\hat{\tilde{\mathsf{X}}}$, random variable taking values from $\{-,+\}$, and $\tilde{\mathsf{X}}_l$ is independent Bernoulli random variable at level $l$ after expansion.

Here, we reduce the original problem to a set of source coding
subproblems over levels $-L_1$ to $L_2$. Similar to the channel coding
case analyzed above, if $\hat{\mathsf{X}}\overset{d.}{\to}\mathsf{X}$,
and $\hat{\tilde{\mathsf{X}}}\overset{d.}{\to}\tilde{\mathsf{X}}^*$,
as $L_1$, $L_2\to\infty$, then the achieved rate distortion pair
approximates the original one. Note that, in general, the decomposition may not be sufficiently close for most of the sources, and the distribution for the estimate may not be sufficiently approximated. These situations add more distortion and result in a gap from the theoretical limit.

%%%%%%%%%%%%%%%%%%%%%%%%%%%%%%%%%%%%%%%%%%%%%%%%%%%%%%%%%%%%%%%%%%%%%%%%%%%%%%

\subsection{Exponential Source Coding Problem Setup}
In this section, a particular lossy compression example is introduced to illustrate the effectiveness of expansion source coding.
Consider an i.i.d. exponential source sequence $\mathsf{X}_1,\ldots, \mathsf{X}_n$, i.e., omitting index $i$,  each variable has a pdf given by
\begin{equation}
f_{\mathsf{X}}(x)=\lambda e^{-\lambda x},\quad x\geq 0,\nonumber
\end{equation}
where $\lambda^{-1}$ is the mean of $\mathsf{X}$. Distortion measure of concern is the ``one-sided error distortion'' given by
\begin{equation}
d(x_{1:n},\tilde{x}_{1:n})=\left\{\begin{array}{ll}
\frac1n\sum\limits_{i=1}^n(x_i-\tilde{x}_i),&\textrm{if } x_{1:n}\succcurlyeq\tilde{x}_{1:n} ,\\
\infty,&\textrm{otherwise,}\end{array}\right.\nonumber
\end{equation}
where $\succcurlyeq$ indicates comparison of vectors element-wise (each element should be greater than the other).
This setup is equivalent to the one in \cite{Verdu:Exponential96}, where another distortion measure is considered.

\begin{lemma}[\cite{Verdu:Exponential96}]\label{lem:ExpSC_Rate_Distortion}
The rate distortion function for an exponential source with the one-sided error distortion is given by
\begin{align}
R(D)=\left\{\begin{array}{ll}
-\log  (\lambda D), &0\leq D\leq \frac{1}{\lambda},\\
0,&D>\frac{1}{\lambda}.
\end{array}
\right.
\end{align}
Moreover, the optimal conditional distribution to achieve the limit is given by
\begin{align}
f^*_{\mathsf{X}|\tilde{\mathsf{X}}}(x|\tilde{x})=\frac{1}{D} e^{- (x-\tilde{x})/D},\quad x\geq\tilde{x}\geq0.\label{equ:ExpSC_Optimal_Conditional_Distribution}
\end{align}
\end{lemma}
\begin{proof}
Proof is given in \cite{Verdu:Exponential96}, and it is based on the observation that among the ensemble of all probability density functions with positive support set and mean constraint, exponential distribution maximizes the differential entropy. By designing a test channel from $\tilde{\mathsf{X}}$ to $\mathsf{X}$, with additive noise distributed as exponential with parameter $1/D$, both the infimum mutual information and optimal conditional distribution can be characterized. Details can be found in Appendix~\ref{app:ExpSC_Rate_Distortion}.
\end{proof}

%%%%%%%%%%%%%%%%%%%%%%%%%%%%%%%%%%%%%%%%%%%%%%%%%%%%%%%%%%%%%%%%%%%%%%%%%%%%%%

\subsection{Expansion Coding for Exponential Source}

Using Lemma \ref{lem:Exponential_Expansion}, we reconstruct the exponential distribution by a set of discrete Bernoulli random variables. In particular, the expansion of exponential source over levels ranging from $-L_1$ to $L_2$ can be expressed as
\begin{equation}
\hat{\mathsf{X}}_i=\sum_{l=-L_1}^{L_2}2^l\mathsf{X}_{i,l},\quad i=1,2,\ldots,n,\nonumber
\end{equation}
where $\mathsf{X}_{i,l}$ are Bernoulli random variables with parameter
\begin{align}
p_{l}\triangleq\textrm{Pr}\{\mathsf{X}_{i,l}=1\}=\frac{1}{1+e^{\lambda 2^l}}.\label{equ:ExpSC_Source_Parameter}
\end{align}
This expansion perfectly approximates exponential source by letting $L_1,L_2\to\infty$. Consider a similar expansion of the source estimate, i.e.,
\begin{equation}
\hat{\tilde{\mathsf{X}}}_i=\sum_{l=-L_1}^{L_2}2^l \tilde{\mathsf{X}}_{i,l},\quad i=1,2,\ldots,n, \nonumber
\end{equation}
where $\tilde{\mathsf{X}}_{i,l}$ is the resulting Bernoulli random variable with parameter $\tilde{p}_l\triangleq\textrm{Pr}\{\tilde{\mathsf{X}}_{i,l}=1\}$.
Utilizing the expansion method, the original problem of coding for a continuous source can be translated to a problem of coding for a set of independent binary sources. This transformation, although seemingly obvious, is valuable as one can utilize powerful coding schemes over discrete sources to achieve rate distortion limits with low complexity. In particular, we design two schemes for the binary source coding problem at each level.

\subsubsection{Coding with one-sided distortion}
In order to guarantee $x\geq \tilde{x}$, we formulate each level as a binary source coding problem under the binary one-sided distortion constraint: $d_{\textrm{O}}(x_l,\tilde{x}_l)=\bold{1}_{\{x_l> \tilde{x}_l\}}$.
Denoting the distortion at level $l$ as $d_l$, an asymmetric test channel (Z-channel) from $\tilde{\mathsf{X}}_{l}$ to $\mathsf{X}_{l}$ can be constructed, where
\begin{align}
&\textrm{Pr}\{\mathsf{X}_l=1|\tilde{\mathsf{X}}_l=0\}=\frac{d_l}{1-p_l+d_l},\nonumber\\
&\textrm{Pr}\{\mathsf{X}_l=0|\tilde{\mathsf{X}}_l=1\}=0.\nonumber
\end{align}
Based on this, we have $p_l-\tilde{p}_l=d_l$, and the achievable rate at a single level $l$ is given by
\begin{equation}
R_{\textrm{Z},l}=H(p_l)-(1-p_l+d_l)H\left(\frac{d_l}{1-p_l+d_l}\right).
\end{equation}
Due to the decomposability property as stated previously, the coding scheme considered is over a set of parallel discrete levels indexed by $l=-L_1,\ldots,L_2$. Thus, by adopting rate distortion limit achieving codes over each level, expansion coding scheme readily achieves the following result.
\begin{theorem}
\label{thm:ExpSC_Achievable_Rate_Z}
For an exponential source, expansion coding achieves the rate distortion pair given by
\emph{
\begin{align}
R_1&=\sum_{l=-L_1}^{L_2}R_{\textrm{Z},l},\label{equ:ExpSC_Achievable_Rate_R1}\\
D_1&=\sum_{l=-L_1}^{L_2}2^ld_l+2^{-L_2+1}/\lambda^2+2^{-L_1-1},\label{equ:ExpSC_Achievable_Rate_D1}
\end{align}}
for any $L_1,L_2>0$, and $d_l\in[0,0.5]$ for $l\in\{-L_1,\cdots,L_2\}$, where $p_l$ is given by \eqref{equ:ExpSC_Source_Parameter}.
\end{theorem}
\begin{proof}
See Appendix~\ref{app:ExpSC_Achievable_Rate_Z}.
\end{proof}

Note that, the last
two terms in \eqref{equ:ExpSC_Achievable_Rate_D1} are a result of the truncation and vanish in the limit of large number of levels. In later parts of this section, we characterize the number of levels required in order to bound the resulting distortion within a constant gap.

\subsubsection{Successive encoding and decoding}
\begin{figure}[t]
 \centering
 \includegraphics[scale=1]{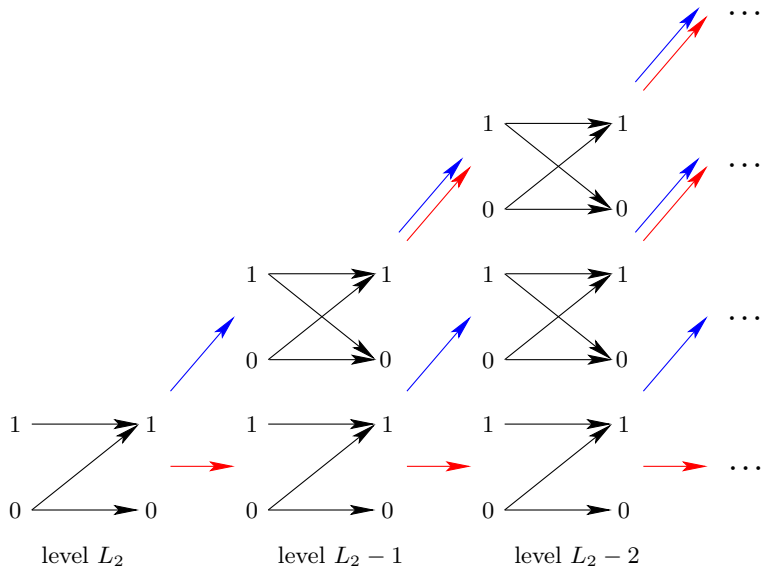}
 \caption{{\bf Illustration of successive encoding and decoding.} Encoding and decoding start from the highest level. A lower level is modeled as one-side distortion (test channel is Z-channel) if and only if estimates in all higher levels are decoded as equal to the source. In this illustration, red arrows represent for decoded as equal, while blue ones represent for decoded as unequal.
}
\label{fig:ExpSC_Successive_Coding}
\end{figure}

Note that it is not necessary to make sure $x_l\geq \tilde{x}_l$ for every $l$ to guarantee $x\geq \tilde{x}$. To this end, we introduce successive coding scheme, where encoding and decoding start from the highest level $L_2$ to the lowest. At a certain level, if all higher levels are encoded as equal to the source, then we must model this level as binary source coding with the one-sided distortion. Otherwise, we formulate this level as binary source coding with the symmetric distortion (see Figure~\ref{fig:ExpSC_Successive_Coding} for an illustration of this successive coding scheme). In particular, for the later case, the distortion of concern is Hamming distortion, i.e. $d_{\textrm{H}}(x_l,\tilde{x}_l)=\bold{1}_{\{x_l\neq \tilde{x}_l\}}$.
Denoting the equivalent distortion at level $l$ as $d_l$, i.e. $\mathbb{E}[\mathsf{X}_l-\tilde{X}_l]=d_l$, then the symmetric test channel from $\hat{\mathsf{X}}_l$ to $\mathsf{X}_l$ is modeled as
\begin{equation}
\textrm{Pr}\{\mathsf{X}_l=1|\hat{\mathsf{X}}_l=0\}=\textrm{Pr}\{\mathsf{X}_l=0|\tilde{\mathsf{X}}_l=1\}
=\frac{d_l}{1-2p_l+2d_l}.\nonumber
\end{equation}
Hence, the achievable rate at level $l$ is given by
\begin{equation}
R_{\textrm{X},l}=H(p_l)-H\left(\frac{d_l}{1-2p_l+2d_l}\right).
\end{equation}
Based on these, we have the following achievable result:
\begin{theorem}
\label{thm:ExpSC_Achievable_Rate_X}
For an exponential source, applying successive coding, expansion coding achieves the rate distortion pairs
\emph{
\begin{align}
R_2&=\sum_{l=-L_1}^{L_2}\left[\alpha_l R_{\textrm{Z},l} +\left(1-\alpha_l\right)R_{\textrm{X},l}\right],\label{equ:ExpSC_Achievable_Rate_R2}\\
D_2&=\sum_{l=-L_1}^{L_2}2^ld_l+2^{-L_2+1}/\lambda^2+2^{-L_1-1},\label{equ:ExpSC_Achievable_Rate_D2}
\end{align}}
for any $L_1,L_2>0$, and $d_l\in[0,0.5]$ for $l\in\{-L_1,\cdots,L_2\}$. Here, $p_l$ is given by \eqref{equ:ExpSC_Source_Parameter}, and the values of $\alpha_l$ are determined by:
\begin{itemize}
\item For $l=L_2$,
\begin{equation}
\alpha_{L_2}=1;\nonumber
\end{equation}
\item For $l<L_2$,
\begin{equation}
\alpha_l= \prod_{k=l+1}^{L_2}(1-d_k).\nonumber
\end{equation}
\end{itemize}
\end{theorem}
\begin{proof}
See Appendix~\ref{app:ExpSC_Achievable_Rate_X}.
\end{proof}

In this sense, the achievable pairs in both theorems are given by optimization problems over a set of parameters $\{d_{-L_1},\ldots,d_{L_2}\}$. However, the problems are not convex, so an effective theoretical analysis may not be performed here for the optimal solution. But, by a heuristic choice of $d_l$, we can still get a good performance. Inspired by the fact that the optimal scheme models noise as exponential with parameter $1/D$ in the test channel, we design $d_l$ as the expansion parameter from this distribution, i.e., we consider
\begin{eqnarray}
d_l=\frac{1}{1+e^{2^l/D}}.\label{equ:ExpSC_Distortion_Parameter}
\end{eqnarray}

We note that higher levels get higher priority and
lower distortion with this choice, which is consistent with the intuition. This choice of $d_l$ may not guarantee any optimality, although simulation results imply that this can be an approximately optimal solution. In the following, we show that the proposed expansion coding scheme achieves within a constant gap to the rate distortion function (at each distortion value).
\begin{theorem}
\label{thm:ExpSC_Main_Result}
For any $D\in[0,1/\lambda]$, there exists a constant $c>0$ (independent of $\lambda$ and $D$), such that if
\begin{itemize}
\item $L_1\geq -\log D$;
\item $L_2\geq -\log \lambda^2D$,
\end{itemize}
then, with a choice of $d_l$ as in \eqref{equ:ExpSC_Distortion_Parameter}, the achievable rates obtained from expansion coding schemes are both within $c=5log(e)$ bits gap to Shannon rate distortion function, i.e.,
\begin{align}
&R_1-R(D_1)\leq c,\nonumber\\
&R_2-R(D_2)\leq c.\nonumber
\end{align}
\end{theorem}
\begin{proof}
See Appendix~\ref{app:ExpSC_Main_Result}.
\end{proof}

\begin{remark}
We remark that the requirement for highest level is much more restricted than channel coding case. The reason is that number of levels should be large enough to approximate both rate and distortion. From the proof in Appendix~\ref{app:ExpSC_Main_Result}, it is evident that $L_2\geq -\log \lambda$ is enough to bound the rate, however, another $-\log \lambda D$ is required to approximate the distortion closely. (If only the relative distortion is considered, these extra levels may not be essential.)
\end{remark}

\begin{remark}
Expansion source coding can be also applied to other source statistics. For instance, for Laplacian (two-sided symmetric exponential) sources, the proposed coding scheme can still approximately approach to the Shannon rate distortion limit with a small constant gap \cite{si2014lossy}, where the sign bit of Laplacian is considered separately and encoded perfectly with $1$ bit, and each expanded level is encoded with Hamming distortion, for low distortion regime.
\end{remark}
%%%%%%%%%%%%%%%%%%%%%%%%%%%%%%%%%%%%%%%%%%%%%%%%%%%%%%%%%%%%%%%%%%%%%%%%%%%%%%

\subsection{Numerical Results}
Numerical results showing achievable rates along with the rate distortion limit are given in Fig.~\ref{fig:ExpSC_Achievable_Rate}. It is evident that both forms of  expansion coding perform within a constant gap of the limit over the whole distortion region, which outperforms existing scalar quantization technique, especially in the low distortion regime. Since samples are independent, the simulations for vector quantization are expected to be close to scalar quantization and omitted in this result.

Theorem~\ref{thm:ExpSC_Main_Result} shows that this gap is bounded by a constant. Here, numerical results show that the gap is not necessarily as wide as predicted by the analysis. Especially for the low distortion region, the gap is numerically found to correspond to $0.24$ bits and $0.43$ bits for each coding scheme respectively.

\begin{figure}[t]
 \centering
 \includegraphics[scale=1]{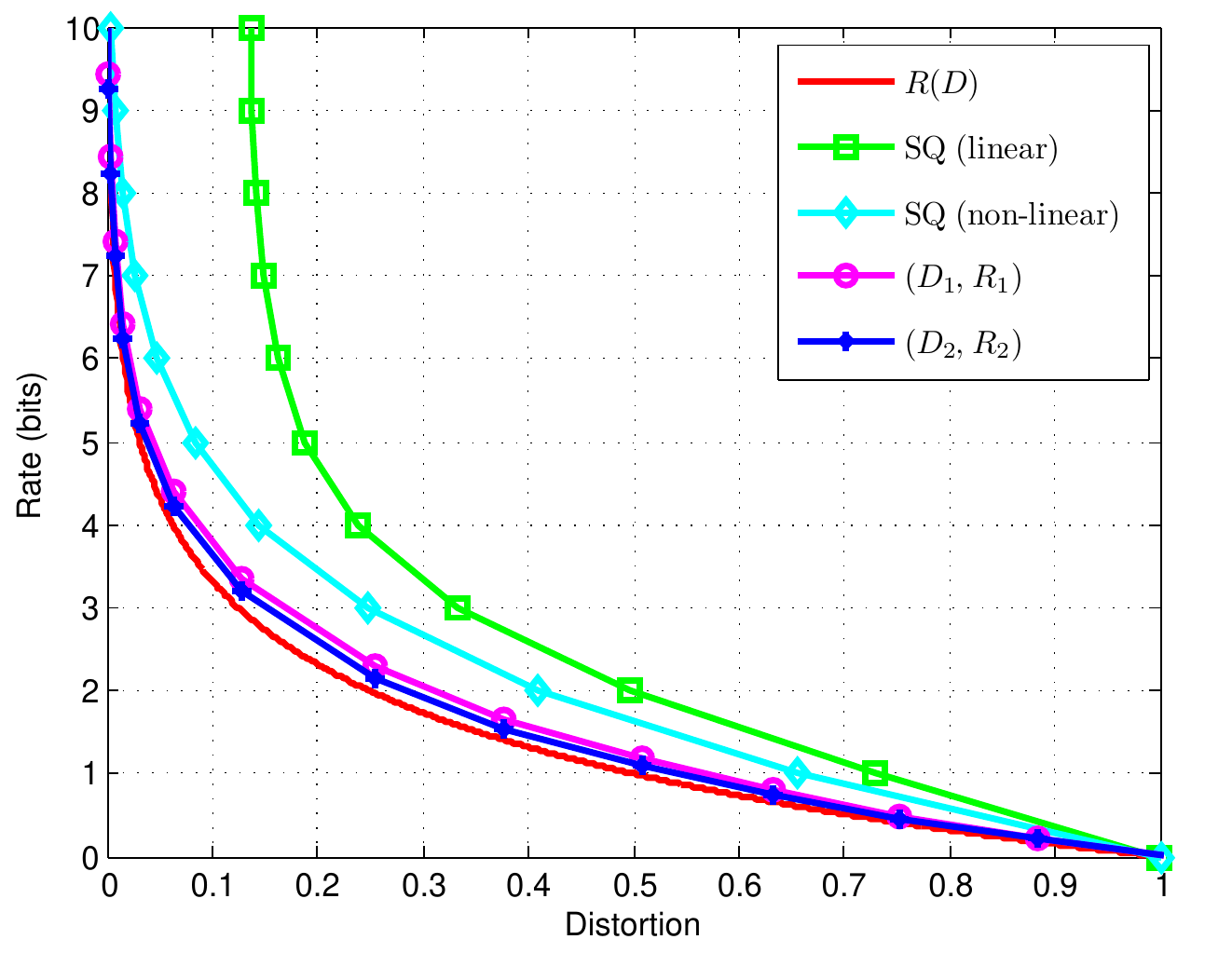}
 \caption{{\bf Achievable rate distortion pairs using expansion coding for exponential distribution ($\lambda=1$) with one-sided error distortion.} $R(D)$ (red-solid) is rate distortion limit; $(D_1,R_1)$ (purple) is given by Theorem~\ref{thm:ExpSC_Achievable_Rate_Z}; $(D_2,R_2)$ (blue) is given by Theorem~\ref{thm:ExpSC_Achievable_Rate_X}. Linear and non-linear scalar quantization methods are simulated for comparison.
}
\label{fig:ExpSC_Achievable_Rate}
\end{figure}

%%%%%%%%%%%%%%%%%%%%%%%%%%%%%%%%%%%%%%%%%%%%%%%%%%%%%%%%%%%%%%%%%%%%%%%%%%%%%%
%%%%%%%%%%%%%%%%%%%%%%%%%%%%%%%%%%%%%%%%%%%%%%%%%%%%%%%%%%%%%%%%%%%%%%%%%%%%%%

\section{Conclusion}
\label{sec:Conclusion}

In this paper, the method of expansion coding is proposed to construct
good codes for analog channel coding and source coding. With a perfect
or approximate decomposition of channel noise and original sources, we
consider coding over independent parallel representations, thus
providing a foundation for reducing the original problems to a set of
parallel simpler subproblems. In particular, via expansion channel
coding, we consider coding over $q$-ary channels for each expansion
level. This approximation of the original channel together with
capacity achieving codes for each level (to reliably transmit messages
over each channel constructed) and Gallager's method (to achieve
desired communication rates for each channel) allow for constructing
near-capacity achieving codes for the original channel. Similarly, we
utilize expansion source coding to adopt discrete source codes that
achieve the rate distortion limit on each level after expansion, and design codes achieving near-optimal performance. Theoretical analysis and numerical results are provided to detail performance guarantees of the proposed expansion coding scheme.

One significant benefit from expansion coding is coding complexity. As indicated in
theoretical analysis, approximately $-2\log\epsilon +\log \textrm{SNR}$ and $-2\log (\lambda^2 D)$ number of levels are sufficient for the channel coding and source coding schemes respectively. Thus, by choosing ``good'' low complexity optimal codes within each level (such as polar codes \cite{Arikan:Channel08,Korada:Source10}), the overall complexity of the coding scheme can be made small for the original continuous-valued channel coding and source coding problems.

Although the discussion in this paper focuses on AEN channels as well
as exponential sources, expansion coding scheme is a more general
framework and its applications are not limited to such
scenarios. Towards this end, any channel noise and any source with
decomposable distribution could fit into the range of expansion
coding. Moreover, the idea of expansion could also be generalized to
network information theory, where it can play a role similar to deterministic models \cite{Avestimehr:Wireless11}. The expanded
channels are not completely deterministic in our case; they possess
different noise levels, which may enable the construction of precise
models for network analysis.

%%%%%%%%%%%%%%%%%%%%%%%%%%%%%%%%%%%%%%%%%%%%%%%%%%%%%%%%%%%%%%%%%%%%%%%%%%%%%%
%%%%%%%%%%%%%%%%%%%%%%%%%%%%%%%%%%%%%%%%%%%%%%%%%%%%%%%%%%%%%%%%%%%%%%%%%%%%%%

\bibliographystyle{IEEEtran}

\begin{thebibliography}{10}
\providecommand{\url}[1]{#1}
\csname url@samestyle\endcsname
\providecommand{\newblock}{\relax}
\providecommand{\bibinfo}[2]{#2}
\providecommand{\BIBentrySTDinterwordspacing}{\spaceskip=0pt\relax}
\providecommand{\BIBentryALTinterwordstretchfactor}{4}
\providecommand{\BIBentryALTinterwordspacing}{\spaceskip=\fontdimen2\font plus
\BIBentryALTinterwordstretchfactor\fontdimen3\font minus
  \fontdimen4\font\relax}
\providecommand{\BIBforeignlanguage}[2]{{%
\expandafter\ifx\csname l@#1\endcsname\relax
\typeout{** WARNING: IEEEtran.bst: No hyphenation pattern has been}%
\typeout{** loaded for the language `#1'. Using the pattern for}%
\typeout{** the default language instead.}%
\else
\language=\csname l@#1\endcsname
\fi
#2}}
\providecommand{\BIBdecl}{\relax}
\BIBdecl

\bibitem{Shannon:IT48}
C.~E. Shannon, ``A mathematical theory of communication,'' \emph{Bell System
  Technical Journal}, vol.~27, no.~3, pp. 379--423, Jul. 1948.

\bibitem{Hamming:code50}
R.~W. Hamming, ``Error detecting and error correcting codes,'' \emph{Bell
  System Technical Journal}, vol.~29, no.~2, pp. 147--160, Apr. 1950.

\bibitem{Golay:code49}
M.~J.~E. Golay, ``Notes on digital coding,'' \emph{Proceeding of the Institute
  of Radio Engineers}, vol.~37, no.~6, pp. 657--657, Jun. 1949.

\bibitem{Muller:code54}
D.~E. Muller, ``Application of boolean algebra to switching circuit design and
  to error detection,'' \emph{IRE Transactions on Electronic Computers},
  vol.~3, no.~3, pp. 6--12, Sep. 1954.

\bibitem{Reed:code54}
I.~S. Reed, ``A class of multiple-error-correcting codes and the decoding
  scheme,'' \emph{IRE Transactions on Information Theory}, vol.~4, no.~4, pp.
  38--49, Sep. 1954.

\bibitem{Reed:RScode60}
I.~S. Reed and G.~Solomon, ``Polynomial codes over certain finite fields,''
  \emph{Journal of the Society for Industrial \& Applied Mathematics}, vol.~8,
  no.~2, pp. 300--304, Jun. 1960.

\bibitem{Conway:Lattice88}
J.~H. Conway and N.~J.~A. Sloane, \emph{Sphere packings, lattices and
  groups}.\hskip 1em plus 0.5em minus 0.4em\relax New York:~Springer, 1988.

\bibitem{Williams:Coding1983}
F.~J. McWilliams and N.~J.~A. Sloane, \emph{The theory for error-correcting
  codes}.\hskip 1em plus 0.5em minus 0.4em\relax North-Holland, 1983.

\bibitem{Forney:convolutional70}
J.~G.~D.~Forney, ``Convolutional codes {I}: {A}lgebraic structure,'' \emph{IEEE
  Transactions on Information Theory}, vol.~16, no.~6, pp. 720--738, Nov. 1970.

\bibitem{Forney:Contatenated66}
------, \emph{Concatenated codes}.\hskip 1em plus 0.5em minus 0.4em\relax
  Cambridge: {MIT} press, 1966.

\bibitem{Berrou:Turbo93}
C.~Berrou and A.~Glavieux, ``Near optimum error correcting coding and decoding:
  {T}urbo-codes,'' \emph{{IEEE} Transactions on Communications}, vol.~44,
  no.~10, pp. 1064--1070, Oct. 1996.

\bibitem{Mackay:LDPC95}
D.~J.~C. MacKay and R.~M. Neal, ``Good codes based on very sparse matrices,''
  in \emph{Cryptography and Coding}.\hskip 1em plus 0.5em minus 0.4em\relax
  Springer, 1995, pp. 100--111.

\bibitem{Sipser:LDPC96}
M.~Sipser and D.~A. Spielman, ``Expander codes,'' \emph{IEEE Transactions on
  Information Theory}, vol.~42, no.~6, pp. 1710--1722, Nov. 1996.

\bibitem{Arikan:Channel08}
E.~Ar{\i}kan, ``Channel polarization: {A} method for constructing
  capacity-achieving codes for symmetric binary-input memoryless channels,''
  \emph{IEEE Transactions on Information Theory}, vol.~55, no.~7, pp.
  3051--3073, Jul. 2009.

\bibitem{Perry:Spinal11}
J.~Perry, H.~Balakrishnan, and D.~Shah, ``Rateless spinal codes,'' in
  \emph{Proc. of the 10th ACM Workshop on Hot Topics in Networks (HotNets
  2011)}, New York City, New York, USA, Nov. 2011, pp. 1--6.

\bibitem{Balakrishnan:Spinal12}
H.~Balakrishnan, P.~Iannucci, D.~Shah, and J.~Perry, ``De-randomizing
  {S}hannon: The design and analysis of a capacity-achieving rateless code,''
  \emph{arXiv:1206.0418}, Jun. 2012.

\bibitem{Cover:IT1991}
T.~M. Cover and J.~A. Thomas, \emph{Elements of information theory}.\hskip 1em
  plus 0.5em minus 0.4em\relax John Wiley \& Sons, 1991.

\bibitem{Viterbi:Trellis74}
A.~J. Viterbi and J.~K. Omura, ``Trellis encoding of memoryless discrete-time
  sources with a fidelity criterion,'' \emph{IEEE Transactions on Information
  Theory}, vol.~20, no.~3, pp. 325--332, May 1974.

\bibitem{Matsunaga:LDPC2003}
Y.~Matsunaga and H.~Yamamoto, ``A coding theorem for lossy data compression by
  {LDPC} codes,'' \emph{IEEE Transactions on Information Theory}, vol.~49,
  no.~9, pp. 2225--2229, Sep. 2003.

\bibitem{Wainwright:LDGM2010}
M.~J. Wainwright, E.~Maneva, and E.~Martinian, ``Lossy source compression using
  low-density generator matrix codes: {A}nalysis and algorithms,'' \emph{IEEE
  Transactions on Information Theory}, vol.~56, no.~3, pp. 1351--1368, Mar.
  2010.

\bibitem{Korada:Source10}
S.~B. Korada and R.~L. Urbanke, ``Polar codes are optimal for lossy source
  coding,'' \emph{IEEE Transactions on Information Theory}, vol.~56, no.~4, pp.
  1751--1768, Apr. 2010.

\bibitem{Zamir:Lattice09}
R.~Zamir, ``Lattices are everywhere,'' in \emph{Proc. 2009 Information Theory
  and Applications Workshop (ITA 2009)}, San Diego, California, USA, Feb. 2009,
  pp. 392--421.

\bibitem{Costello:Channel07}
J.~D.~J.~Costello and J.~G.~D.~Forney, ``Channel coding: The road to channel
  capacity,'' \emph{Proceedings of the IEEE}, vol.~95, no.~6, pp. 1150--1177,
  Jun. 2007.

\bibitem{Forney:Sphere00}
J.~G.~D.~Forney, M.~D. Trott, and S.~Y. Chung, ``Sphere-bound-achieving coset
  codes and multilevel coset codes,'' \emph{IEEE Transactions on Information
  Theory}, vol.~46, no.~3, pp. 820--850, May 2000.

\bibitem{Avestimehr:Wireless11}
A.~Avestimehr, S.~Diggavi, and D.~Tse, ``Wireless network information flow: {A}
  deterministic approach,'' \emph{IEEE Transactions on Information Theory},
  vol.~57, no.~4, pp. 1872--1905, Apr. 2011.

\bibitem{Abbe:Polar11}
E.~Abbe and A.~Barron, ``Polar coding schemes for the {AWGN} channel,'' in
  \emph{Proc. 2011 IEEE International Symposium on Information Theory ({ISIT}
  2011)}, St.~Petersburg, Russia, Jul. 2011, pp. 194--198.

\bibitem{Abbe:Polar12}
E.~Abbe and E.~Telatar, ``Polar codes for the $m$-user multiple access
  channel,'' \emph{IEEE Transactions on Information Theory}, vol.~58, no.~8,
  pp. 5437--5448, Aug. 2012.

\bibitem{Seidl:Polar13}
M.~Seidl, A.~Schenk, C.~Stierstorfer, and J.~B. Huber, ``Polar-coded
  modulation,'' \emph{IEEE Transactions on Communications}, vol.~61, no.~10,
  pp. 4108--4119, Oct. 2013.

\bibitem{Verdu:Exponential96}
S.~Verd{\'u}, ``The exponential distribution in information theory,''
  \emph{Problems of Information Transmission}, vol.~32, no.~1, pp. 100--111,
  Jan. 1996.

\bibitem{Martinez:Communication11}
A.~Martinez, ``Communication by energy modulation: The additive exponential
  noise channel,'' \emph{IEEE Transactions on Information Theory}, vol.~57,
  no.~6, pp. 3333--3351, Jun. 2011.

\bibitem{LeGoff:Capacity11}
S.~Y. {Le~Goff}, ``Capacity-approaching signal constellations for the additive
  exponential noise channel,'' \emph{IEEE Wireless Communications Letters},
  vol.~1, no.~4, pp. 320--323, Aug. 2012.

\bibitem{Gallager:Information68}
R.~G. Gallager, \emph{Information theory and reliable communication}.\hskip 1em
  plus 0.5em minus 0.4em\relax John Wiley \& Sons, 1968.

\bibitem{gray1998quantization}
R.~M. Gray and D.~L. Neuhoff, ``Quantization,'' \emph{IEEE Transactions on
  Information Theory}, vol.~44, no.~6, pp. 2325--2383, Oct. 1998.

\bibitem{Baron:MCMC12}
D.~Baron and T.~Weissman, ``An {MCMC} approach to universal lossy compression
  of analog sources,'' \emph{IEEE Transactions on Signal Processing}, vol.~60,
  no.~10, pp. 5230--5240, Oct. 2012.

\bibitem{si2014lossy}
H.~Si, O.~O. Koyluoglu, and S.~Vishwanath, ``Lossy compression of exponential
  and laplacian sources using expansion coding,'' in \emph{Proc. 2014 {IEEE}
  International Symposium on Information Theory ({ISIT} 2014)}, Honolulu,
  Hawaii, USA, Jul. 2014, pp. 3052--3056.

\bibitem{Marsaglia:Random71}
G.~Marsaglia, ``Random variables with independent binary digits,'' \emph{The
  Annals of Mathematical Statistics}, vol.~42, no.~6, pp. 1922--1929, Dec.
  1971.

\end{thebibliography}

%%%%%%%%%%%%%%%%%%%%%%%%%%%%%%%%%%%%%%%%%%%%%%%%%%%%%%%%%%%%%%%%%%%%%%%%%%%%%%
%%%%%%%%%%%%%%%%%%%%%%%%%%%%%%%%%%%%%%%%%%%%%%%%%%%%%%%%%%%%%%%%%%%%%%%%%%%%%%

\appendices

%%%%%%%%%%%%%%%%%%%%%%%%%%%%%%%%%%%%%%%%%%%%%%%%%%%%%%%%%%%%%%%%%%%%%%%%%%%%%%

\section{Proof of Lemma \ref{lem:Exponential_Expansion}}
\label{app:Exponential_Expansion_Proof}
The ``if'' part follows by extending the one given in
\cite{Marsaglia:Random71}, which considers the expansion
of a truncated exponential random variable.
We show the result by calculating the
moment generating function of $\mathsf{B}$.
Using the assumption that $\{\mathsf{B}_l\}_{l\in\mathbb{Z}}$ are
mutually independent, we have
\begin{align}
M_{\mathsf{B}}(t)  =\mathbb{E}\left[e^{t{\mathsf{B}}}\right]
        %=\mathbb{E}\left[e^{t\sum\limits_{l=-\infty}^{\infty}2^l {\mathsf{B}}_l}\right]
        =\prod_{l=-\infty}^{\infty}\mathbb{E}\left[e^{t2^l {\mathsf{B}}_l}\right].\nonumber
\end{align}
Noting that $\mathsf{B}_l$ are Bernoulli random variables, we have
\begin{align}
\mathbb{E}\left[e^{t2^l {\mathsf{B}}_l}\right]=
\frac{e^{t2^l}}{1+e^{\lambda 2^l}}+\left(1-\frac{1}{1+e^{\lambda 2^l}}\right)
=\frac{1+e^{(t-\lambda ) 2^l}}{1+e^{-\lambda 2^l}}.\nonumber
\end{align}
Then, using the fact that for any constant $\alpha\in\mathbb{R}$,
\begin{align}
\prod_{l=0}^{n}(1+e^{\alpha 2^l})=\frac{1-e^{2^{n+1}\alpha}}{1-e^{\alpha}},\nonumber
\end{align}
we can obtain the following for $t<\lambda$,
\begin{align}
\prod_{l=0}^{\infty} \mathbb{E}\left[e^{t2^l {\mathsf{B}}_l}\right]
=\lim_{n\rightarrow\infty}\prod_{l=0}^{n}\frac{1+e^{(t-\lambda ) 2^l}}{1+e^{-\lambda 2^l}}
= \frac{1-e^{-\lambda}}{1-e^{t-\lambda}}.\label{equ:Exponential_Expansion_Proof_Part1}
\end{align}
Similarly, for the negative part, we have
\begin{align}
\prod_{l=-n}^{-1}(1+e^{\alpha 2^l})=\frac{1-e^{\alpha}}{1-e^{\alpha2^{-n}}},\nonumber
\end{align}
which further implies that
\begin{align}
\prod_{l=-\infty}^{-1}\mathbb{E}\left[e^{t2^l {\mathsf{B}}_l}\right]
=\lim_{n\rightarrow\infty}\frac{1-e^{t-\lambda}}
{1-e^{(t-\lambda)2^{-n}}}\frac{1-e^{-\lambda2^{-n}}}{1-e^{-\lambda}}
=\frac{\lambda(1-e^{t-\lambda})}{(\lambda-t)(1-e^{-\lambda})}.
\label{equ:Exponential_Expansion_Proof_Part2}
\end{align}

Thus, finally for any $t<\lambda$, combining equations~\eqref{equ:Exponential_Expansion_Proof_Part1} and~\eqref{equ:Exponential_Expansion_Proof_Part2}, we get
\begin{align}
M_{\mathsf{B}}(t)=\frac{\lambda}{\lambda-t}.\nonumber
\end{align}
The observation that this is the moment generation function for an exponentially distributed random variable with parameter $\lambda$ concludes the proof.

The independence relationships between levels in ``only if'' part can be simply verified using memoryless property of the exponential distribution. Here, we just
need to show the parameter for Bernoulli random variable at each level. Observe that for any $l\in\mathbb{Z}$,
\begin{align}
\textrm{Pr}\{\mathsf{B}_l=1\}=\textrm{Pr}\{\mathsf{B}\in\cup_{k\in\mathbb{N}^+}[2^l(2k-1),2^l(2k))\}
=\sum_{k\in\mathbb{N}^+}\textrm{Pr}\{2^l(2k-1)\leq \mathsf{B}<2^l(2k)\}.
\label{equ:Exponential_Expansion_Proof_Part3}
\end{align}
Using cdf of exponential distribution, we obtain
\begin{align}
\textrm{Pr}\{2^l(2k-1)\leq \mathsf{B}<2^l(2k)\}=e^{-\lambda2^l(2k-1)}-e^{-\lambda2^l(2k)}
=e^{-\lambda2^l(2k)}\left(e^{\lambda 2^l}-1\right).\nonumber
\end{align}
Using this in \eqref{equ:Exponential_Expansion_Proof_Part3}, we have
\begin{align}
\textrm{Pr}\{{\mathsf{B}}_l=1\}=\sum_{k=1}^{\infty}e^{-\lambda2^l(2k)}\left(e^{\lambda 2^l}-1\right)=\frac{1}{e^{\lambda 2^l}+1}. \nonumber
\end{align}

%%%%%%%%%%%%%%%%%%%%%%%%%%%%%%%%%%%%%%%%%%%%%%%%%%%%%%%%%%%%%%%%%%%%%%%%%%%%%%

\section{Proof of Lemma \ref{lem:AEN_Entropy_Bound}}
\label{app:AEN_Entropy_Bound}

From \eqref{equ:AEN_Noise_Parameter}, and $\eta\triangleq\log E_{\mathsf{Z}}$, we have
\begin{equation}
q_l=\frac{1}{1+e^{2^l/E_{\mathsf{Z}}}}=\frac{1}{1+e^{2^{l-\eta}}}.\nonumber
\end{equation}
By definition of entropy, we obtain
\begin{align}
H(q_l)  &=-q_l\log q_l-(1-q_l)\log (1-q_l)\nonumber\\
        &=-\frac{1}{1+e^{2^{l-\eta}}}\log \frac{1}{1+e^{2^{l-\eta}}}-\frac{e^{2^{l-\eta}}}{1+e^{2^{l-\eta}}}\log \frac{e^{2^{l-\eta}}}{1+e^{2^{l-\eta}}}.\nonumber
\end{align}

When $l\leq\eta$, we obtain a lower bound as
\begin{align}
H(q_l)  &=\frac{1}{1+e^{2^{l-\eta}}}\log \left(1+e^{2^{l-\eta}}\right)+\frac{e^{2^{l-\eta}}}{1+e^{2^{l-\eta}}}\log \left(\frac{1+e^{2^{l-\eta}}}{e^{2^{l-\eta}}}\right)\nonumber\\
        &=\log \left(1+e^{2^{l-\eta}}\right) -\frac{e^{2^{l-\eta}}}{1+e^{2^{l-\eta}}}\log e \cdot2^{l-\eta}\nonumber\\
        &\stackrel{(a)}{>}\log(1+1)-\log e\cdot 2^{l-\eta}\nonumber\\
        &=1-\log e \cdot2^{l-\eta},\nonumber
\end{align}
where $(a)$ is due to $e^{2^{l-\eta}}>1$ and $-e^{2^{l-\eta}}/(1+e^{2^{l-\eta}})>-1$.

On the other hand, when $l> \eta$, we have
\begin{align}
H(q_l)  &=\frac{1}{1+e^{2^{l-\eta}}}\log \left(1+e^{2^{l-\eta}}\right)+\frac{e^{2^{l-\eta}}}{1+e^{2^{l-\eta}}}\log \left(\frac{1+e^{2^{l-\eta}}}{e^{2^{l-\eta}}}\right)\nonumber\\
        &\stackrel{(b)}{<}\frac{1}{1+e^{2^{l-\eta}}}\log\left(2e^{2^{l-\eta}}\right)+ \log \left(1+e^{-2^{l-\eta}}\right)\nonumber\\
        &\stackrel{(c)}{<}\frac{1+2^{l-\eta}\cdot\log e}{1+e^{2^{l-\eta}}}+ e^{-2^{l-\eta}}\cdot\log e\nonumber\\
        &\stackrel{(d)}{<}\frac{1+2^{l-\eta}\cdot\log e}{1+1+2^{l-\eta}+2^{2({l-\eta})}/2}+ \frac{\log e}{1+2^{l-\eta}}\nonumber\\
        &\stackrel{(e)}{<}2^{\eta-l+1}\cdot\log e + 2^{\eta-l}\cdot\log e\nonumber\\
	&=3\log e\cdot2^{\eta-l} ,\nonumber
\end{align}
where
\begin{itemize}
\item[$(b)$] is from $1<e^{2^{l-\eta}}$ and $e^{2^{l-\eta}}/(1+e^{2^{l-\eta}})<1$;
\item[$(c)$] is from $\log (1+\alpha)< \alpha \log e$ for any $0<\alpha<1$;
\item[$(d)$] is from $e^{\alpha} > 1+\alpha+\alpha^2/2>1+\alpha$ for any $\alpha>0$;
\item[$(e)$] is from $1+2^{l-\eta}\cdot\log e<(2+2^{l-\eta}+2^{2({l-\eta})}/2)(2^{{\eta-l}+1}\cdot\log e)$
and $1<(1+2^{l-\eta})2^{{\eta-l}}$ for any $l$ and $\eta$.
\end{itemize}

%%%%%%%%%%%%%%%%%%%%%%%%%%%%%%%%%%%%%%%%%%%%%%%%%%%%%%%%%%%%%%%%%%%%%%%%%%%%%%

\section{Proof of Lemma \ref{lem:AEN_Equivalent_Entropy_Bound}}
\label{app:AEN_Equivalent_Entropy_Bound}

By definition, $\tilde{q}_l=c_l\otimes q_l$, so its behavior is closely related to carries. Note that for any $l$, we have
\begin{align}
\tilde{q}_l=c_l(1-q_l)+q_l(1-c_l)
=q_l+c_l(1-2q_l)
\geq q_l,\nonumber
\end{align}
where the last inequality holds due to $q_l<1/2$ and $c_l\geq 0$.
Then, for $l\leq \eta$, we have
\begin{equation}
H(\tilde{q}_l)\geq H(q_l)>1-\log e\cdot2^{l-\eta},\nonumber
\end{equation}
where the first inequality holds due to monotonicity of entropy, and the last inequality is due to \eqref{equ:AEN_Entropy_Bound2} in  Lemma~\ref{lem:AEN_Entropy_Bound}. For the $l>\eta$ part, we need to characterize carries first. We have the following assertion:
\begin{equation}
c_l<2^{\eta-l+1}-\frac{2}{1+e^{2^{l-\eta}}},\quad \textrm{for } \; l>\eta,\label{equ:AEN_Equivalent_Noise_Proof1}
\end{equation}
and the proof is based on the following induction analysis. For $l=\eta+1$, this is simply true, because $c_l<1/2$ for any $l$. Assume \eqref{equ:AEN_Equivalent_Noise_Proof1}
 is true for level ${l>\eta}$, then at the ${l}+1$ level, we have
\begin{align}
c_{l+1} &=p_lq_l(1-c_l)+p_l(1-q_l)c_l+(1-p_l)q_lc_l+p_lq_lc_l\nonumber\\
        &=p_l(c_l+q_l-2q_lc_l)+q_lc_l\nonumber\\
        &\stackrel{(a)}{<}\frac12(c_l+q_l-2q_lc_l)+q_lc_l\nonumber\\
        &=\frac12(c_l+q_l)\nonumber\\
        &\stackrel{(b)}{<}\frac12\left(2^{\eta-l+1}-\frac{2}{1+e^{2^{l-\eta}}}
        +\frac{1}{1+e^{2^{l-\eta}}}\right)\nonumber\\
        &\stackrel{(c)}{<}2^{-({l-\eta}+1)+1}-\frac{2}{1+e^{2^{{l-\eta}+1}}},\nonumber
\end{align}
where
\begin{itemize}
\item[$(a)$] is due to $p_l<1/2$ and $c_l+q_l-2q_lc_l=c_l(1-2q_l)+q_l>0$;
\item[$(b)$] is due to the assumption \eqref{equ:AEN_Equivalent_Noise_Proof1} for level $l$;
\item[$(c)$] is due to the fact that $1/[2(1+e^{2^{l-\eta}})]>2/(1+e^{2^{{l-\eta}+1}})$ holds for any $l>\eta$. To this end, the assertion also holds for level $l+1$, and this completes the proof of \eqref{equ:AEN_Equivalent_Noise_Proof1}.
\end{itemize}

Using \eqref{equ:AEN_Equivalent_Noise_Proof1}, we obtain that for any $l>\eta$
\begin{align}
\tilde{q}_l &=q_l+c_l(1-2q_l)\nonumber\\
            &<\frac{1}{1+e^{2^{l-\eta}}}+\left(2^{\eta-l+1}-\frac{2}{1+e^{2^{l-\eta}}}\right)
            \left(1-\frac{2}{1+e^{2^{l-\eta}}}\right)\nonumber\\
            &=2^{\eta-l+1}-\frac{1+2^{\eta-l+2}}{1+e^{2^{l-\eta}}}
            +\frac{4}{(1+e^{2^{l-\eta}})^2}\nonumber\\
            &<2^{\eta-l+1},\label{equ:AEN_Equivalent_Noise_Proof2}
\end{align}
where the last inequality holds due to $(1+2^{\eta-l+2})(1+e^{2^{l-\eta}})>4$ for any $l>\eta$.

Finally, we obtain
\begin{align}
H(\tilde{q}_l)  &\stackrel{(d)}{<} H(2^{\eta-l+1})\nonumber\\
                &=-2^{\eta-l+1}\log (2^{\eta-l+1})-(1-2^{\eta-l+1})\log(1-2^{\eta-l+1})\nonumber\\
                &\stackrel{(e)}{<}({l-\eta}-1)\cdot 2^{\eta-l+1}+(1-2^{\eta-l+1})2^{\eta-l+1}2\log e\nonumber\\
                &\stackrel{(f)}{<} (l-\eta)\cdot 2^{\eta-l+1}+(l-\eta)2^{\eta-l+1}2\log e\nonumber\\
                &<3\log e\cdot ({l-\eta})\cdot 2^{\eta-l+1}\nonumber\\
                &=6\log e\cdot ({l-\eta})\cdot 2^{\eta-l},\nonumber
\end{align}
where
\begin{itemize}
\item[$(d)$] is from \eqref{equ:AEN_Equivalent_Noise_Proof2} and the monotonicity of entropy;
\item[$(e)$] is from $-\log (1-\alpha)<2\alpha\log e$ for any $\alpha\leq 1/2$;
\item[$(f)$] is from $1-2^{\eta-l+1}<l-\eta$ for any $l>\eta$.
\end{itemize}

From the proof, the information we used for $p_l$ is that $p_l<1/2$,
so this bound holds uniformly for any SNR.

%%%%%%%%%%%%%%%%%%%%%%%%%%%%%%%%%%%%%%%%%%%%%%%%%%%%%%%%%%%%%%%%%%%%%%%%%%%%%%

\section{Proof of Theorem \ref{thm:AEN_Achivable_Rate_Main_Result}}
\label{app:AEN_Achivable_Rate_Main_Result}

We first prove that $\hat{R}_2$ achieves capacity. Denote $\xi=\log E_{\mathsf{X}}$ and $\eta= \log E_{\mathsf{Z}}$. Then, we have an important observation that
\begin{align}
p_{l}=\frac{1}{1+e^{2^{l}/2^{\xi}}}=q_{l+\eta-\xi},\label{equ:AEN_Rate_Gap_Proof1}
\end{align}
which shows that channel input is a shifted version of noise with respect to expansion levels (see Fig.~\ref{fig:AEN_Expanded_Level} for intuition). Based on this, we have
\begin{align}
\hat{R}_2   &=\sum_{l=-L_1}^{L_2}\left[ H(p_l\otimes q_l)-H(q_l) \right]\nonumber\\
            &\stackrel{(a)}{\geq}\sum_{l=-L_1}^{L_2}\left[ H(p_l)-H(q_l)\right]\nonumber\\
            &\stackrel{(b)}{=}\sum_{l=-L_1}^{L_2}\left[ H(q_{l+\eta-\xi})-H(q_l)\right]\nonumber\\
            &=\sum_{l=-L_1+\eta-\xi}^{L_2+\eta-\xi}H(q_l) -\sum_{l=-L_1}^{L_2}H(q_l)\nonumber\\
            &=\sum_{l=-L_1+\eta-\xi}^{-L_1-1}H(q_l) -\sum_{l=L_2+\eta-\xi+1}^{L_2}H(q_l)\nonumber\\
            &\stackrel{(c)}{\geq}\sum_{l=-L_1+\eta-\xi}^{-L_1-1}\left[1-2^{l-\eta}\log e \right] -\sum_{l=L_2+\eta-\xi+1}^{L_2}2^{\eta-l}3\log e\nonumber\\
            &\stackrel{(d)}{\geq}(\xi-\eta)-2^{-L_1-\eta}\log e -2^{-L_2+\xi}3\log e \nonumber\\
            &\stackrel{(e)}{\geq}\log \left(\frac{E_{\mathsf{X}}}{E_{\mathsf{Z}}}\right)-\epsilon\log e-3\epsilon\log e\epsilon\nonumber\\
            &\stackrel{(f)}{\geq}\log \left(1+\frac{E_{\mathsf{X}}}{E_{\mathsf{Z}}}\right)-\frac{E_{\mathsf{Z}}}{E_{\mathsf{X}}}\log e-4\epsilon\log e\nonumber\\
            &\stackrel{(g)}{\geq}\log \left(1+\frac{E_{\mathsf{X}}}{E_{\mathsf{Z}}}\right)-5\epsilon\log e,\label{equ:AEN_Rate_Gap_Proof2}
\end{align}
where
\begin{itemize}
\item[$(a)$] is due to $p_l\otimes q_l =p_l(1-q)+(1-p_l)q_l\geq p_l$, and monotonicity of entropy;
\item[$(b)$] follows from \eqref{equ:AEN_Rate_Gap_Proof1};
\item[$(c)$] follows from \eqref{equ:AEN_Entropy_Bound1} and \eqref{equ:AEN_Entropy_Bound2} in Lemma~\ref{lem:AEN_Entropy_Bound};
\item[$(d)$] holds as
\begin{equation}
\sum_{l=-L_1+\eta-\xi}^{-L_1-1} 2^{l-\eta}\leq \sum_{l=-\infty}^{-L_1-1} 2^{l-\eta} =2^{-L_1-\eta},\nonumber
\end{equation}
and
\begin{equation}
\sum_{l=L_2}^{L_2+\eta-\xi+1} 2^{\eta-l}\leq \sum_{l=\infty}^{L_2+\eta-\xi+1} 2^{\eta-l} =2^{-L_2+\xi};\nonumber
\end{equation}
\item[$(e)$] is due to the assumptions that $L_1\geq -\log \epsilon-\eta$, and $L_2\geq -\log \epsilon+\xi$;
\item[$(f)$] is due to the fact that
\begin{align}
\log\left(1+\frac{E_{\mathsf{X}}}{E_{\mathsf{Z}}}\right)-\log\left(\frac{E_{\mathsf{X}}}{E_{\mathsf{Z}}}\right)
=\log\left(1+\frac{E_{\mathsf{Z}}}{E_{\mathsf{X}}}\right)\leq  \frac{E_{\mathsf{Z}}}{E_{\mathsf{X}}} \log e ,\nonumber
\end{align}
as $\log(1+\alpha)\leq \alpha \log e $ for any $\alpha\geq 0$;
\item[$(g)$] is due to the assumption that $\textrm{SNR}=E_{\mathsf{X}}/E_{\mathsf{Z}}\geq 1/\epsilon$.
\end{itemize}

Next, we show the result for $\hat{R}_1$. Observe that
\begin{align}
\hat{R}_1   &=\sum_{l=-L_1}^{L_2}\left[ H(p_l\otimes \tilde{q}_l)-H(\tilde{q}_l) \right]\nonumber\\
            &\stackrel{(h)}{\geq} \sum_{l=-L_1}^{L_2}\left[ H(p_l\otimes q_l)-H(\tilde{q}_l) \right]\nonumber\\
            &=\sum_{l=-L_1}^{L_2}\left[ H(p_l\otimes q_l)-H(q_l) \right]+\sum_{l=-L_1}^{L_2}\left[ H(q_l)-H(\tilde{q}_l) \right]\nonumber\\
            &=\hat{R}_2-\sum_{l=-L_1}^{L_2}\left[ H(\tilde{q}_l)-H(q_l) \right]\nonumber\\
            &\stackrel{(i)}{\geq}\hat{R}_2-\sum_{-L_1}^{\eta}\left[1-\left(1-2^{l-\eta}\log e \right)\right]-\sum_{\eta+1}^{L_2}\left[6(l-\eta)2^{-l+\eta}\log e-0\right]\nonumber\\
            &=\hat{R}_2-\sum_{-L_1}^{\eta} 2^{l-\eta} \log e -\sum_{\eta+1}^{L_2}6(l-\eta)2^{-l+\eta}\log e\nonumber\\
            &\stackrel{(j)}{\geq}\hat{R}_2- 2\log e -12\log e\nonumber\\
            &\stackrel{(k)}{\geq}\log \left(1+\frac{E_{\mathsf{X}}}{E_{\mathsf{Z}}}\right)-5\epsilon\log e-14\log e,\nonumber
\end{align}
where
\begin{itemize}
\item[$(h)$] is due to $\tilde{q}_l\geq q_l$, which further implies $p_l\otimes \tilde{q}_l\geq p_l\otimes q_l$;
\item[$(i)$] follows from \eqref{equ:AEN_Entropy_Bound2} and \eqref{equ:AEN_Equivalent_Entropy_Bound1}, together with the fact that $H(\tilde{q}_l)\leq 1$ and $H(q_l)\geq0$ for any $l$;
\item[$(j)$] follows from the observations that
\begin{equation}
\sum_{l=-L_1}^{\eta} 2^{l-\eta}\log e \leq \log e\cdot 2^{-\eta}\cdot\sum_{l=-\infty}^{\eta} 2^{l}=2\log e,\nonumber
\end{equation}
and
\begin{equation}
\sum_{l=\eta+1}^{L_2}6(l-\eta)2^{-l+\eta} \log e\leq 6\log e\cdot\sum_{l=\eta+1}^{\infty}(l-\eta)\cdot2^{-l+\eta}=12\log e;\nonumber
\end{equation}
\item[$(k)$] is due to \eqref{equ:AEN_Rate_Gap_Proof2}.
\end{itemize}

Thus, choosing $c=19\log e$ completes the proof. Note that, in the course of providing these upper bounds, the actual gap might be enlarged. The actual value of the gap is much smaller (e.g., as shown in Fig.~\ref{fig:AEN_Achievable_Rate}, numerical result for the capacity gap is around $1.72$ bits).

%%%%%%%%%%%%%%%%%%%%%%%%%%%%%%%%%%%%%%%%%%%%%%%%%%%%%%%%%%%%%%%%%%%%%%%%%%%%%%

\section{Proof of Lemma~\ref{lem:ExpSC_Rate_Distortion}}
\label{app:ExpSC_Rate_Distortion}

Note that the maximum entropy theorem implies that
the distribution maximizing differential entropy over all probability densities $f$ on support set $\mathbb{R}^+$ satisfying
\begin{equation}
\int_{0}^{\infty}f(x)dx=1,\quad\int_{0}^{\infty}f(x)xdx=1/\lambda,\nonumber
\end{equation}
is exponential distribution with parameter $\lambda$. Based on this result, in order to satisfy $\mathbb{E}[d(\mathsf{X},\tilde{\mathsf{X}})]\leq D$, where $d(\mathsf{X},\tilde{\mathsf{X}})=\infty$ for $\mathsf{X}<\tilde{\mathsf{X}}$, we have to restrict $\mathsf{X}\geq \hat{\mathsf{X}}$ with probability $1$. To this end,  we have
\begin{align}
I(\mathsf{X};\tilde{\mathsf{X}})&=h(\mathsf{X})-h(\mathsf{X}|\tilde{\mathsf{X}})\nonumber\\
            &=\log(e/\lambda)-h(\mathsf{X}-\tilde{\mathsf{X}}|\tilde{\mathsf{X}})\nonumber\\
            &\geq \log(e/\lambda)-h(\mathsf{X}-\tilde{\mathsf{X}})\nonumber\\
            &\geq \log(e/\lambda)-\log(e\mathbb{E}[\mathsf{X}-\tilde{\mathsf{X}}])\nonumber\\
            &\geq \log(e/\lambda)-\log(eD)\nonumber\\
            &=-\log (\lambda D).\nonumber
\end{align}

To achieve this bound, we need $\mathsf{X}-\tilde{\mathsf{X}}$ to be exponentially distributed and independent with $\tilde{\mathsf{X}}$ as well. Accordingly, we can consider a test channel from $\tilde{\mathsf{X}}$ to $\mathsf{X}$ with additive noise $\mathsf{Z}=\mathsf{X}-\tilde{\mathsf{X}}$ distributed as exponential with parameter $1/D$, which gives the conditional distribution given by (\ref{equ:ExpSC_Optimal_Conditional_Distribution}).

%%%%%%%%%%%%%%%%%%%%%%%%%%%%%%%%%%%%%%%%%%%%%%%%%%%%%%%%%%%%%%%%%%%%%%%%%%%%%%

\section{Proof of Theorem~\ref{thm:ExpSC_Achievable_Rate_Z}}
\label{app:ExpSC_Achievable_Rate_Z}

Due to decomposability of exponential distribution, the levels after expansion are independent, hence, the achievable rate in this theorem is obtained by additions of individual rates. On the other hand, for the calculation of distortion, we have
\begin{align}
D_1 &=\mathbb{E}[\mathsf{X}-\hat{\tilde{\mathsf{X}}}]\nonumber\\ &=\mathbb{E}\left[\sum_{l=-\infty}^{\infty}2^l\mathsf{X}_l-\sum_{l=-L_1}^{L_2}2^l\tilde{\mathsf{X}}_l\right]\nonumber\\
    &\overset{(a)}{=}\sum_{l=-L_1}^{L_2}2^ld_l +\sum_{l=L_2+1}^{\infty} 2^lp_l+\sum_{l=-\infty}^{-L_1-1}2^lp_l\nonumber\\
    &\overset{(b)}{\leq} \sum_{l=-L_1}^{L_2}2^ld_l +\sum_{l=L_2+1}^{\infty} 2^{-l+1}/\lambda^2 +\sum_{l=-\infty}^{-L_1-1}2^{l-1}\nonumber\\
    &\overset{(c)}{=} \sum_{l=-L_1}^{L_2}2^ld_l + 2^{-L_2+1}/\lambda^2+2^{-L_1-1},\nonumber
\end{align}
where
\begin{itemize}
\item[$(a)$] follows from $p_l-\tilde{p}_l=d_l$;
\item[$(b)$] follows from
\begin{equation}
p_l=\frac{1}{1+e^{\lambda 2^l}}\leq \frac{1}{1+(1+\lambda 2^l+\lambda^2 2^{2l}/2)}\leq\frac{1}{\lambda^2 2^{2l}/2}=2^{-2l+1}/\lambda^2,\nonumber
\end{equation}
and $p_l<1/2$ for any $l$;
\item[$(c)$] follows from
\begin{equation}
\sum_{l=L_2+1}^{\infty} 2^{-l}=2^{-L_2},\textrm{ and } \sum_{l=-\infty}^{-L_1-1}2^l=2^{-L_1}.\nonumber
\end{equation}
\end{itemize}

%%%%%%%%%%%%%%%%%%%%%%%%%%%%%%%%%%%%%%%%%%%%%%%%%%%%%%%%%%%%%%%%%%%%%%%%%%%%%%

\section{Proof of Theorem~\ref{thm:ExpSC_Achievable_Rate_X}}
\label{app:ExpSC_Achievable_Rate_X}

By the design of coding scheme, if all higher levels are decoded as equivalence, then they must be encoded with one-sided distortion. Recall that for Z-channel, we have
\begin{equation}
\textrm{Pr}\{\mathsf{X}_l\neq \tilde{\mathsf{X}}_l\}=\textrm{Pr}\{\mathsf{X}_l=1,\tilde{\mathsf{X}}_l=0\}=d_l.\nonumber
\end{equation}
Hence, due to independence of expanded levels,
\begin{equation}
\alpha_l=\prod_{k=l+1}^{L_2}(1-d_k).\nonumber
\end{equation}
Then, at each level, the achievable rate is $R_{\textrm{Z},l}$ with probability $\alpha_l$ and is $R_{\textrm{X},l}$ otherwise. From this, we obtain the expression of $R_2$ given by the theorem. On the other hand, since in both cases we have $p_l-\tilde{p}_l=d_l$, the form of distortion remains the same.

%%%%%%%%%%%%%%%%%%%%%%%%%%%%%%%%%%%%%%%%%%%%%%%%%%%%%%%%%%%%%%%%%%%%%%%%%%%%%%

\section{Proof of Theorem~\ref{thm:ExpSC_Main_Result}}
\label{app:ExpSC_Main_Result}

Denote $\gamma = -\log \lambda$, and $\xi = -\log D$. Then, from $D\leq 1/\lambda$, we have
\begin{equation}
\gamma+\xi\geq 0.\label{equ:ExpSC_Main_Proof1}
\end{equation}
By noting that $p_l$ and $d_l$ are both expanded parameters from exponential distribution, we have
\begin{align}
&p_l = \frac{1}{1+e^{\lambda 2^l}}=\frac{1}{1+e^{2^{l-\gamma}}},\nonumber\\
&d_l = \frac{1}{1+e^{2^l/D}}=\frac{1}{1+e^{2^{l+\xi}}}.\nonumber
\end{align}
Hence, $p_l$ is shifted version of $d_l$ (analog to the channel coding case), i.e.,
\begin{equation}
d_l = p_{l+\gamma+\xi}.\label{equ:ExpSC_Main_Proof2}
\end{equation}
Using this relationship, we obtain
\begin{align}
\sum_{l=-L_1}^{L_2}\left[H(p_l)-H(d_l)\right]
&\overset{(a)}{=}\sum_{l=-L_1}^{L_2}H(p_l)-\sum_{l=-L_1}^{L_2}H(p_{l+\gamma+\xi})\nonumber\\
&=\sum_{l=-L_1}^{L_2}H(p_l)-\sum_{l=-L_1+\gamma+\xi}^{L_2+\gamma+\xi}H(p_{l})\nonumber\\
&\overset{(b)}{=}\sum_{l=-L_1}^{-L_1+\gamma+\xi-1}H(p_l)-\sum_{l=L_2+1}^{L_2+\gamma+\xi}H(p_{l})\nonumber\\
&\overset{(c)}{\leq} \gamma+\xi,\label{equ:ExpSC_Main_Proof3}
\end{align}
where
\begin{itemize}
\item[$(a)$] follows from \eqref{equ:ExpSC_Main_Proof2};
\item[$(b)$] follows from \eqref{equ:ExpSC_Main_Proof1} and theorem assumptions;
\item[$(c)$] follows from $0\leq H(p_l)\leq 1$ for any $l$.
\end{itemize}

From the expression of $R_1$, we have
\begin{align}
R_1 &= \sum_{l=-L_1}^{L_2} \left[H(p_l)-(1-p_l+d_l)H\left(\frac{d_l}{1-p_l+d_l}\right)\right]\nonumber\\
&=\sum_{l=-L_1}^{L_2}\left[H(p_l)-H(d_l)\right] +\sum_{l=-L_1}^{L_2} \left[H(d_l)-(1-p_l+d_l)H\left(\frac{d_l}{1-p_l+d_l}\right)\right]\nonumber\\
&\leq \gamma+\xi +\sum_{l=-L_1}^{L_2} \left[H(d_l)-(1-p_l+d_l)H\left(\frac{d_l}{1-p_l+d_l}\right)\right],\label{equ:ExpSC_Main_Proof4}
\end{align}
where \eqref{equ:ExpSC_Main_Proof3} is used to obtain the last inequality. It remains to bound
\begin{align}
\Delta_l    &\triangleq H(d_l)-(1-p_l+d_l)H\left(\frac{d_l}{1-p_l+d_l}\right)\nonumber\\
            &=(1-p_l)\log(1-p_l)-(1-d_l)\log(1-d_l)-(1-p_l+d_l)\log(1-p_l+d_l).\nonumber
\end{align}
For this, two cases are considered:
\begin{enumerate}
\item[1)] For $l\leq -\xi$, $d_l$ and $p_l$ are close and both tend to $0.5$. More precisely, we have
    \begin{align}
    \Delta_l    &\overset{(d)}{\leq} -(1-p_l+d_l)\log(1-p_l+d_l)\nonumber\\
                &\overset{(e)}{\leq} 2(p_l-d_l)\log e\nonumber\\
                &\overset{(f)}{\leq} 2\left[\frac{1}{2}-\left(\frac12-2^{l+\xi-1}\right)\right]\log e \nonumber\\
                &=2^{l+\xi}\log e  ,\label{equ:ExpSC_Main_Proof5}
    \end{align}
    where
    \begin{itemize}
    \item[$(d)$] follows from the fact that $(1-\alpha)\log(1-\alpha)$ is a decreasing function over $[0,0.5]$, hence, $(1-p_l)\log(1-p_l)\leq(1-d_l)\log(1-d_l)$;
    \item[$(e)$] follows from the observation that $-(1-\alpha)\log(1-\alpha)\leq 2\log e\cdot\alpha$ for any $\alpha\in[0,0.5]$;
    \item[$(f)$] follows from the fact that $p_l\leq 0.5$ and
    \begin{equation}
    d_l = \frac{1}{1+e^{2^{l+\xi}}}\geq \frac{1}{1+(1+2\cdot 2^{l+\xi})}\geq \frac12-2^{l+\xi-1},\nonumber
    \end{equation}
    where the first inequality is due to $e^{\alpha}\leq 1+2\alpha$ for any $\alpha\in[0,1]$ ($2^{l+\xi}\leq 1$ due to $l\leq -\xi$), and the last inequality holds for any $l$.
    \end{itemize}
    \item[2)] On the other hand, for $l>-\xi$, $d_l$ tends to $0$, so as $1-p_l$ and $1-p_l+d_l$ get close. More precisely, we have
    \begin{align}
    \Delta_l    &\overset{(g)}{\leq} -(1-d_l)\log(1-d_l)\nonumber\\
                &\overset{(h)}{\leq} 2 d_l\log e\nonumber\\
                &\overset{(i)}{\leq} 2^{-l-\xi}\log e  ,\label{equ:ExpSC_Main_Proof6}
    \end{align}
    where
    \begin{itemize}
    \item[$(g)$] follows from the fact $(1-p_l)\log(1-p_l)\leq(1-p_l+d_l)\log(1-p_l+d_l)$;
    \item[$(h)$] follows from the observation that $-(1-\alpha)\log(1-\alpha)\leq 2\alpha\log e$ for any $\alpha\in[0,0.5]$;
    \item[$(i)$] follows from the fact that
    \begin{equation}
    d_l=\frac{1}{1+e^{2^{l+\xi}}}\leq \frac{1}{e^{2^{l+\xi}}}\leq \frac{1}{2\cdot2^{l+\xi}}=2^{-l-\xi-1},\nonumber
    \end{equation}
    where the second inequality holds from $e^{\alpha}\geq 2\alpha$ for any $\alpha>1$ ($2^{l+\xi}>1$ due to $l>-\xi$).
    \end{itemize}
\end{enumerate}
Putting \eqref{equ:ExpSC_Main_Proof5} and \eqref{equ:ExpSC_Main_Proof6} back to \eqref{equ:ExpSC_Main_Proof4}, we have
\begin{align}
R_1 &\leq \gamma+\xi+\log e\cdot\sum_{l=-L_1}^{-\xi}2^{l+\xi}+\log e\cdot\sum_{l=-\xi+1}^{L_2}2^{-l-\xi}\nonumber\\
    &\leq \gamma+\xi+2\log e+\log e\nonumber\\
    &=R(D)+3\log e,\label{equ:ExpSC_Main_Proof7}
\end{align}
where we use the definitions of $\gamma$ and $\xi$, such that $\gamma+\xi=R(D)$.

Finally, using the result from Theorem~\ref{thm:ExpSC_Achievable_Rate_Z} that
\begin{equation}
D\leq D_1\leq \sum_{l=-L_1}^{L_2}2^ld_l + 2^{-L_2+1}/\lambda^2+2^{-L_1-1}\leq D + 2^{-L_2+1}/\lambda^2+2^{-L_1-1},\label{equ:ExpSC_Main_Proof8}
\end{equation}
we obtain
\begin{align}
R(D)    &\overset{(j)}{\leq} R(D_1)+\frac{\log e}{D}(D_1-D)\nonumber\\
        &\overset{(k)}{\leq} R(D_1)+\frac{\log e}{D}(2^{-L_2+1}/\lambda^2+2^{-L_1-1})\nonumber\\
        &\overset{(l)}{\leq} R(D_1)+2.5\log e,\label{equ:ExpSC_Main_Proof9}
\end{align}
where
\begin{itemize}
\item[$(j)$] follows from $R(D)$ is convex such that for any $\beta$ and $\alpha$,
\begin{align}
R(\beta)\geq R(\alpha)+R'(\alpha)(\beta-\alpha),\nonumber
\end{align}
where $R'(\alpha)=-\log e/\alpha$ is the derivative, and setting $\alpha=D$, $\beta=D_1$ completes the proof of this step;
\item[$(k)$] follows from \eqref{equ:ExpSC_Main_Proof8};
\item[$(l)$] follows from theorem assumptions that $L_1\geq -\log D$ and  $L_2\geq -\log \lambda^2 D$.
\end{itemize}
Combining \eqref{equ:ExpSC_Main_Proof9} with \eqref{equ:ExpSC_Main_Proof7}, we have
\begin{equation}
R_1\leq R(D_1)+5.5\log e,\nonumber
\end{equation}
which completes the proof for $R_1$ and $D_1$ by taking $c=5.5\log e$.

For the other part of the theorem, observe that
\begin{align}
H\left( \frac{d_l}{1-2p_l+2d_l} \right)\geq (1-p_l+d_l)H\left( \frac{d_l}{1-p_l+d_l} \right).\nonumber
\end{align}
Hence, for any $-L_1\leq l\leq L_1$, we have $R_{\textrm{X},l}\leq R_{\textrm{Z},l}.$
Thus, we have $R_2\leq R_1$. Combing with the observation that $D_1=D_2$, we have $R_2\leq R(D_2)+5.5\log e$.

%%%%%%%%%%%%%%%%%%%%%%%%%%%%%%%%%%%%%%%%%%%%%%%%%%%%%%%%%%%%%%%%%%%%%%%%%%%%%%

\end{document}